\documentclass[superscriptaddress,
twocolumn,
nofootinbib,
amsmath,amssymb,
aps,
pra,
floatfix
]{revtex4-2}

\usepackage{color}

\usepackage{graphicx}
\usepackage{bm}
\usepackage{amsmath,amsthm}
\usepackage{amssymb}
\usepackage{braket}
\usepackage{etoolbox}
\usepackage{subfig}
\usepackage{hyperref}
\usepackage[dvipsnames]{xcolor}
\usepackage{tikz-cd}
\usepackage[normalem]{ulem}

\newtheorem{lemma}{Lemma}

\newcommand{\Tr}[1]{\mathrm{Tr}}
\newcommand{\sech}[1]{\mathrm{sech}}

\newcommand{\blk}{\color{black}}
\usepackage{ragged2e} 
\begin{document}
\setcounter{page}{1}


\title{Finite-size secret-key rates of discrete modulation continuous-variable quantum key distribution under  Gaussian attacks}

\author{Gabriele Staffieri}
\affiliation{Dipartimento Interateneo di Fisica, Università di Bari, 70126 Bari, Italy}
\affiliation{INFN, Sezione di Bari, 70126 Bari, Italy}

\author{Giovanni Scala}
\affiliation{Dipartimento Interateneo di Fisica, Politecnico di Bari, 70126 Bari, Italy}
\affiliation{INFN, Sezione di Bari, 70126 Bari, Italy}

\author{Cosmo Lupo}
\affiliation{Dipartimento Interateneo di Fisica, Politecnico di Bari, 70126 Bari, Italy}
\affiliation{Dipartimento Interateneo di Fisica, Università di Bari, 70126 Bari, Italy}
\affiliation{INFN, Sezione di Bari, 70126 Bari, Italy}

\begin{abstract}
\noindent
Quantum conditional entropies play a fundamental role in quantum information theory.
In quantum key distribution, they are exploited to obtain reliable lower bounds on the secret-key rates in the finite-size regime, against collective attacks and coherent attacks under suitable assumptions. 
Here we consider continuous-variable communication protocols, where the sender Alice encodes information using a discrete modulation of phase-shifted coherent states, and the receiver Bob decodes by homodyne or heterodyne detection.
We compute the Petz-Rényi and sandwiched Rényi conditional entropies associated with these setups, 
assuming either a passive eavesdropper or one that injects thermal photons into the channel,
who gathers the quantum information leaked through a lossy communication line of known or bounded transmittance.
Whereas our results do not directly provide reliable key-rate estimates, they do represent useful ball-park figures. 
We obtain analytical or semi-analytical expressions that do not require intensive numerical calculations.
These expressions serve as bounds on the key rates that may be tight in certain scenarios.
We compare different estimates, including known bounds that have already appeared in the literature
and
new bounds.
The latter are found to be tighter for very short block sizes.
\end{abstract}

 \maketitle

\section{Introduction}

\noindent
Quantum key distribution (QKD) provides a solution to a fundamental problem in cryptography: the distribution of secret keys among users connected through an insecure communication channel~\cite{QCintro}. Quantum mechanics allows us to achieve this goal by encoding information into quantum states, and with additional resources such as an authenticated classical communication channel. The latter can be authenticated by a pre-shared secret key, making QKD a scheme of secret-key expansion.

Continuous-variable (CV) QKD protocols implement QKD by encoding information in the phase and/or quadrature of the quantum electromagnetic field~\cite{CVrev}. Here, we focus on protocols where such information is encoded by displacing the amplitude of coherent states in phase space.
CV QKD is of particular interest due to its high level of compatibility with existing communication infrastructure, for the scalability and low cost of the required hardware~\cite{Entropyrev}. Proving its security is a challenging task, though, also because the relevant Hilbert spaces have infinite dimensions. In contrast to discrete-variable QKD, proof techniques based on entropic uncertainty relations are either not applicable or yield sub-optimal estimates of the key rates~\cite{Furrer2012,Leverrier2015}.

In recent years, reliable numerical methods to estimate the secret-key rates have been developed and successfully applied to a variety of QKD protocols~\cite{reliableColes,reliableWinick}. They are reliable in the sense that they do not over-estimate the key rates. This approach is particularly suitable to assess the security of CV QKD with discrete modulation (DM), where information is encoded into a finite constellation of coherent states.
Initial works on DM CV QKD focused on the asymptotic limit of infinite channel uses~\cite{Lutk2019,kanitschar2022optimizing} (see also Refs.~\cite{Ghorai2019,Denys2021}).
More recent works have encompassed statistical fluctuations for a finite number of channel uses (finite-size effects), addressing the challenges of obtaining secret keys even when processing blocks of quantum signals of relatively small size~\cite{Lupo2022, FlorianPRXQ2023, AcinQuantum2024, AcinPRA2025}, and provided experimental demonstrations~\cite{TobiasDM2024}.
Indeed, recent mathematical and theoretical developments have significantly advanced the analysis of finite-size effects, especially in the context of coherent attacks within the scenarios encompassed by the entropic accumulation theorems~\cite{EAT2020, improved2ndorder, LiuNP2021, MetgerGEAT, Dupui2023_no_smoothing, George2022,REVdiqkd, MEAT}.

Inspired by these works, here we focus on two distinct models for the channel connecting the sender (Alice) and the receiver (Bob). The first is a pure loss channel, modelling e.g.~attenuation in a characterized optical fiber. This scenario constrains the eavesdropper (Eve) to a particular form of passive attacks: she collects the field that leaks into the environment en route from Alice to Bob. The second model is a thermal noisy channel, which extends the previous description by including the effects of gaussian noise in addition to the losses.
There Gaussian attacks are not sufficient to assess the general security of the QKD protocols, for which one would make minimal or no assumptions on the attack performed by Eve. However, the passive attack model allows us to obtain analytical or semi-analytical expressions for the entropic quantities that characterize the communication scheme, such as the Petz-Rényi and sandwiched Rényi divergences~\cite{Petz1986,Wilde2014,Tomamichel2014,BertaIEEE,BertaJMP}. 
These results give us insights into the relation between different proof techniques and allows us to compare them. Furthermore, our entropic estimates represent useful benchmarks towards the analysis of key rates for unknown channels, providing upper bounds on the key rates given a known, minimal amount of loss and thermal noise expected in the communication channel.
Specifically, we focus on collective attacks in the finite-size regime. We consider two examples of communication protocols: the first is based on encoding by binary phase-shift keying (BPSK) and decoding by homodyne detection, the second is based on encoding by quadrature phase-shift keying (QPSK) and decoding by heterodyne detection. 
We compare different quantum conditional entropies associated with different proof techniques used to bound the finite-size secret-key rate. 
Our study sheds some light on which proof techniques are more efficient across different regimes of loss and block size.
For very small block sizes, we discuss a new approach that has not been explicitly considered by other authors, obtaining new a more favorable bounds on the secret-key rates.
Our work paves the way to a more complete and reliable security analysis in the finite-size regime, which may be developed for CV QKD along the paths indicated by Refs.~\cite{Cai2025, Tan_arXiv2025, Tan_PRXQ} in the context of discrete-variable QKD.
\begin{figure}
    \centering
    \subfloat[BPSK protocol.]{
        \includegraphics[width=0.47\linewidth]{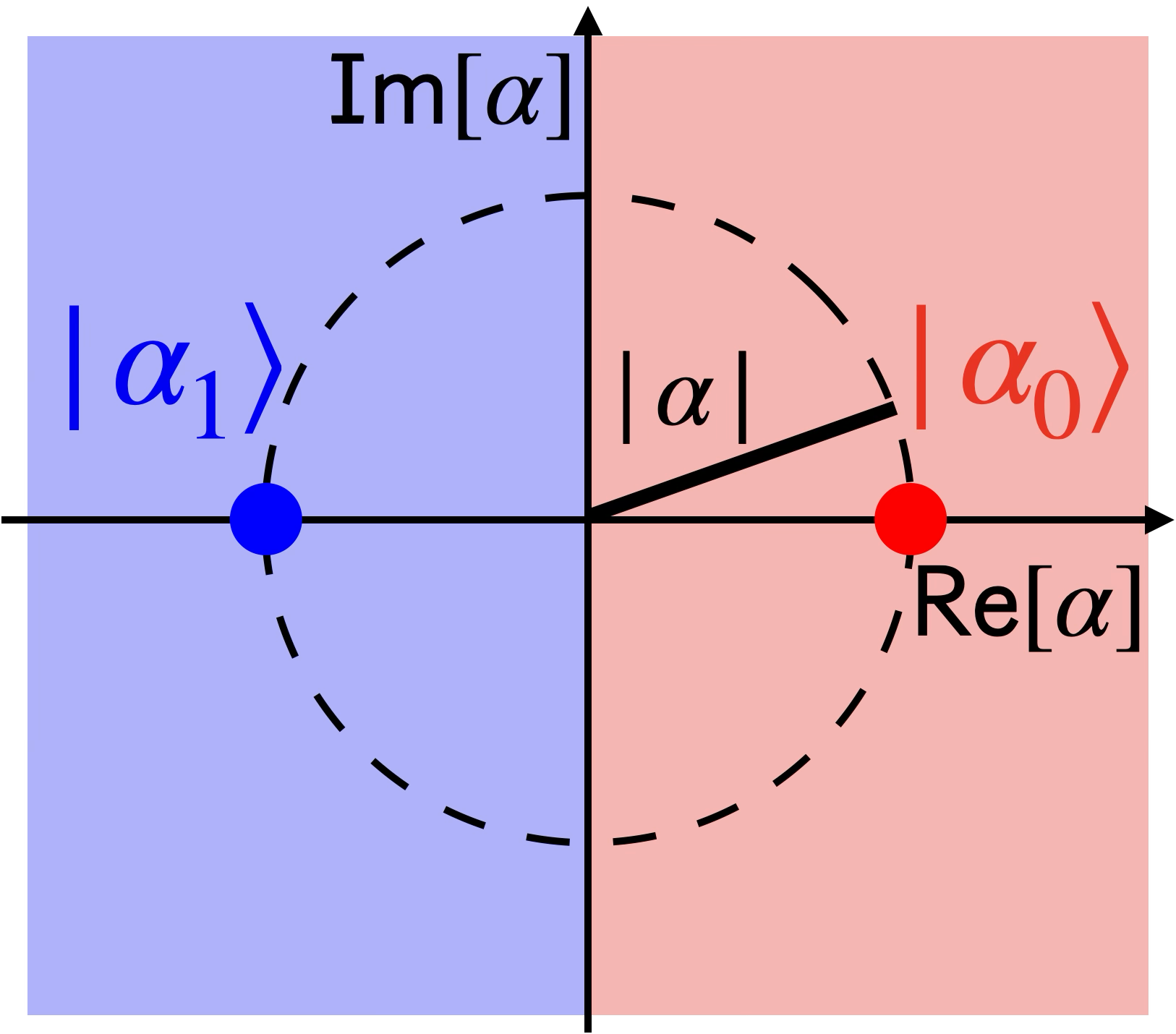}
        \label{bpskspace}
    }
    \subfloat[QPSK protocol.]{
        \includegraphics[width=.47\linewidth]{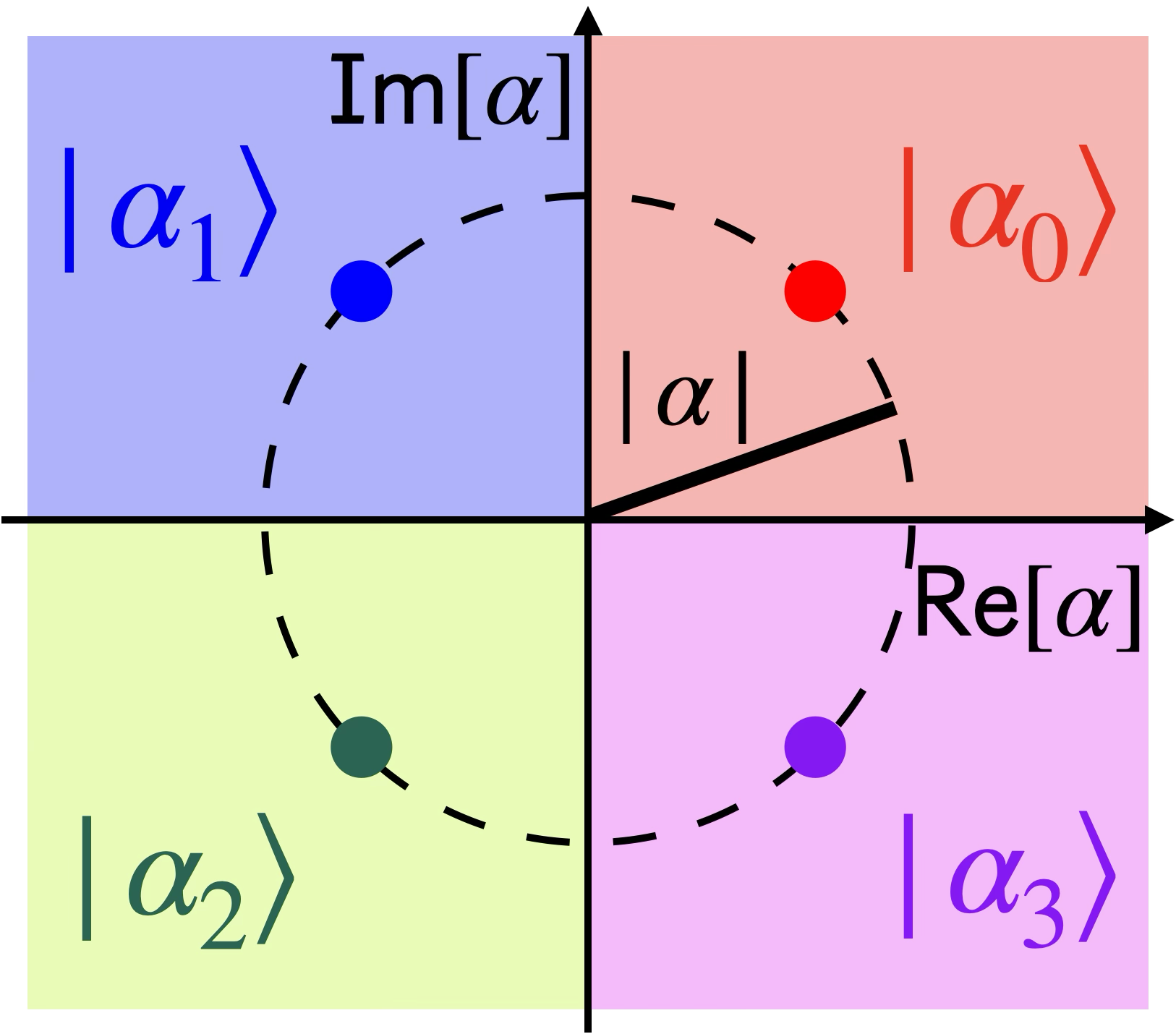}
        \label{QpskScheme}
    }
    \caption{\justifying 
    Phase-space representation of the encoding and decoding routines in BPSK (Fig.~\ref{bpskspace}) and QPSK (Fig.~\ref{QpskScheme}) protocols.}
    \label{fig:overall_PS}
\end{figure}

\section{Min-entropy lower bounds}
\label{sec1}
We assess the secrecy of QKD protocols by exploiting well-known entropic functionals, such as the 
smooth min-entropy (denoted as ${\tilde H_\text{min}^{\uparrow \epsilon}}$),
the Petz-Rényi entropies (denoted as $H_a^\downarrow$ and $H_a^\uparrow$), 
and the sandwiched Rényi conditional entropies (denoted as $\tilde H_a^\downarrow$ and $\tilde H_a^\uparrow$), whose definitions and main properties we briefly review in Appendix~\ref{sec2}. In our analysis we deal with classical-quantum (CQ) states of the form
\begin{align}\label{CQdef}
    \rho_{YE}
    =
    \sum_{y=0}^{N-1} p_y
    \ket{y}_Y \bra{y}
    \otimes 
    \rho_{E|y} \, ,
\end{align}
where $\{\ket{y}_Y\}$, 
with $y=0, \dots , N-1$, 
is an orthonormal basis encoding a classical register $Y$ (Bob's raw key), with associated probability $p_y$,
and $\rho_{E|y}$ is the density matrix representing Eve's state given $y$.
Here we focus on the model of collective attacks, where $n$ copies of the state $\rho_{YE}$ are given.
By applying a suitable randomness extractor on $n$ i.i.d.~instances of $Y$, one can obtain up to $\ell_{\epsilon''}$ nearly-uniformly random bits~\cite{tomamichel2012framework}, with
\begin{align}\label{hash}
    \ell_{\epsilon''} \geq {\tilde H_\text{min}^{\uparrow \epsilon} (Y^n | E^n)} \color{black} - 2 \log{(1/\epsilon')} + 1 \, ,
\end{align} 
where the obtained random bits are $\epsilon''$-close to be uniformly random and uncorrelated with the Eve system $E^n$, and
$\epsilon''=\epsilon+\epsilon'$. 
The conditional min-entropy is bounded by the sandwiched Rényi conditional entropy~\cite{EAT2020}, for any $a>1$~\footnote{As noted by Zhang, Fu, and Knill~\cite{ZhangFuKnill2020}, the restriction \(\alpha\in(1,2]\) used in the proofs of Lemmas~B.4 and~B.10 of Ref.~\cite{EAT2020} is unnecessary: invoking Proposition~6.22 of Tomamichel~\cite{TomamichelBook2016} yields the same type of bound for all \(\alpha>1\). See also Proposition~2.2 in Ref.~\cite{Buscemi2019}. See Fig.~\ref{fig:trans}.}
\begin{align}
    \tilde H_\text{min}^{\uparrow\epsilon}(Y^n|E^n)
     \geq \tilde H^\uparrow_a(Y^n|E^n) - \frac{g(\epsilon)}{a -1}\nonumber\\
     = n \tilde H^\uparrow_a(Y|E) - \frac{g(\epsilon)}{a -1}
    \, ,
    \label{RateSand}
\end{align}
with $g(\epsilon) = - \log{ \left( 1 - \sqrt{1-\epsilon^2 }\right) } \leq \log(2/\epsilon^2)$. 
Our goal is to estimate this lower bound for the BPSK and QPSK protocols.
Similar bounds also hold, where
the other conditional Rényi entropies appear instead of $\tilde H^\uparrow_a$.

On the other hand, the asymptotic equipartition property (AEP) links the conditional smooth min-entropy with the conditional von Neumann entropy \cite{tomamichel2012framework}:
\begin{align}
            \frac{1}{n} \,  \tilde H_\text{min}^{\uparrow\epsilon}(Y^n|E^n)
            \geq
        H(Y|E) - \frac{\delta(\epsilon)}{\sqrt{n}}
         \label{RateAEP}
\end{align}
with $\delta(\epsilon) = 4 \log{ \left( 2 + 
\sqrt{ N } \right) }\,  
\sqrt{
\log{\left( 2/\epsilon^2 \right)}
}$ 
and $N$ is the cardinality of $Y$. 
The bound in Eq.~\eqref{RateAEP} is tight for large $n$ but it becomes loose for relatively small block sizes. {In the finite-size regime, other authors~\cite{AcinQuantum2024,AcinPRA2025}} have exploited a bound of the form 
\begin{equation}\label{boundref18}
    \tilde H_a^\uparrow(Y|E)
    \ge B_a(Y|E) \, ,
\end{equation} 
which holds for $a\in(1,2]$ and bounds the sandwiched Rényi entropy in terms of the von Neumann entropy \cite{improved2ndorder}, where
\begin{align}
    B_a(Y|E) & =  H(Y|E) - \frac{(a-1) \ln{2}}{2}  V(Y|E)\nonumber\\
    & \hspace{2cm} - (a-1)^2 K(a) \, . 
\label{VKBound}
\end{align}
Here $V(A|B):= V(\rho_{AB}||\mathbb I_A\otimes \rho_B)$ is the \textit{conditional entropy variance} defined as
\begin{align}
    V(\rho||\sigma) 
    = \frac{1}{\mathrm{tr}(\rho)}
    \left\{
    \mathrm{tr}
    \left[
    \rho(\log \rho-\log\sigma)^2
    \right]
    \right\}
    -D(\rho||\sigma)^2,
\end{align}
and
\begin{align}
    K(a)
& = \frac{1}{ 6 (2-a)^3 \ln{2}} \,
2^{ (a-1) (- H_a^\downarrow(Y|E) + H(Y|E))}\nonumber\\
&\times
\ln^3
{
\left(
2^{
- H_2^\downarrow (Y|E) + H(Y|E)
}
+e^2
\right)}
\, .
\end{align} 
Using Eq.~\eqref{boundref18} into Eq.~\eqref{RateSand} yields
\begin{align}
            \frac{1}{n} \, \tilde H_\text{min}^{\uparrow\epsilon} (Y^n | E^n)  
            \ge B_a(Y|E)- \frac{g(\epsilon)}{n(a -1)}.
    \label{RateBound}
\end{align}
In the rest of the paper we compare the 
bounds on the smooth min-entropy $\tilde H_{\mathrm{min}}^{\uparrow \epsilon}$ in 
\textit{(i)} Eq.~\eqref{RateSand} (for any $a > 1$), 
\textit{(ii)} Eq.~\eqref{RateAEP}, 
and \textit{(iii)} Eqs.~\eqref{RateBound} (for $a\in(1,2]$), and then use these quantities to estimate the secret-key rates for BPSK and QPSK protocols under the assumption of passive attacks. 
In Sec.~\ref{relatedW} we extend this analysis to include more general Gaussian attacks where 
the eavesdropper injects thermal photons into the channel.

\section{Phase Shift Keying Protocols}
\label{sec:model}
We consider two examples of DM communication protocols where the sender Alice prepares, with uniform probability, quantum states selected from a constellation of $N$ single-mode coherent states, denoted as $|\alpha_x \rangle$ for $x=0,
\dots,N-1$ ~\cite{2statep}. 
In particular, we focus on phase-shift keying (PSK) 
$\{\ket{\alpha_x}\}_{x=0}^{N-1}$,
where the coherent states only differ by their relative phase, $\alpha_x = e^{i 2 x \pi/N } \alpha$ (we assume $\alpha = |\alpha|$): binary PSK (BPSK) for $N=2$, 
and quadrature PSK (QPSK) for $N=4$.
For decoding, the receiver Bob applies either homodyne (for BPSK) or heterodyne (for QPSK) detection (a.k.a.~double homodyne detection) to infer the value of $x$ encoded by Alice.
Consider the quadrature and phase operators, 
$q = (a + a^\dag)/\sqrt{2}$, 
$p = -i (a - a^\dag)/\sqrt{2}$, 
satisfying the canonical commutation relations $[q,p]=i$ (in natural units),
homodyne detection implements (within some approximation) a measurement of either the $q$ or $p$ operator (or a linear combination of them).
If Bob receives the coherent state $|\alpha_x\rangle$, 
by homodyne detection of the quadrature operator $q$ he can learn the real part of the amplitude, since
$\langle q \rangle = \sqrt{2} \, \mathsf{Re} \{ \alpha_x \}$.
For $N=2$, Bob decoding scheme is described by the POVM elements $Q_y$, with $y=0,1$, with $Q_0 = \int_{q \leq 0} |q\rangle \langle q| dq$, and 
$Q_1 = \int_{q > 0} |q\rangle \langle q| dq$, where $|q\rangle$ is the generalized eigenstate of the quadrature operator $q$.
BPSK encoding and decoding are schematically shown in Fig.~\ref{bpskspace}.
By contrast, heterodyne detection implements a joint measurement of both quadrature and phase operators. This is achieved by first splitting the signal with a $50$--$50$ beam splitter. 
Then the output of one branch of the beam splitter is measured with homodyne detection along the $q$ quadrature, and the other branch output is measured along $p$.
This measurement is formally described by a continuous family of POVM (positive-operator-valued measurement) elements, $\{ \Lambda_\gamma \}$, where
$\Lambda_\gamma = \frac{1}{\pi} \, |\gamma\rangle \langle \gamma|$ is proportional to the one-dimensional projector on the coherent state $|\gamma\rangle$ of complex amplitude $\gamma$. For $N=4$ the decoding POVM formally reads as
\begin{align}\label{eq:Lambda}
    \Lambda_y &= 
    \frac{1}{\pi}
    \int_{R_y}\
    \ket{\gamma}_B\bra{\gamma}\,\mathrm{d}^2\gamma\; ,
    \qquad 
    \qquad y=0,\dots,3,\nonumber\\
    R_y &= 
    \left\{
    \gamma \in \mathbb C 
    \bigg| 
    \arg(\gamma) \in 
    \left[
    y\frac{\pi}{2}, (y+1)\frac{\pi}{2}
    \right]
    , |\gamma|>0
    \right\}.    
\end{align}
The encoding and decoding of QPSK is schematically shown in Fig.~\ref{QpskScheme}.

The protocols are parametrized by the amplitude parameter $|\alpha|$ that appears in the state preparation, and therefore in $\rho_{YE}$. 
The optimal bound on the secret-key rate can be obtained following an optimization over both $|\alpha|$ and $a$. For passive attacks, the quantum channel from Alice to Bob is described as pure-loss channel with known transmittance factor $\eta \in [0,1]$.

\paragraph{Pure loss channel---}
The channel models linear loss in well-characterized optical fibers. We may formally represent it as a beam splitter with transmittance $\eta$. 
Linear loss maps coherent states into coherent states, i.e.~given Alice input $|\alpha_x\rangle$, Bob obtains the attenuated coherent state
$|\sqrt{\eta} \, \alpha_x\rangle$.
In this passive-attack scenario, a potential eavesdropper has access to the environment of the channel, which in our case is the output of the reflective arm of the beam splitter. 
Therefore, given Alice input $|\alpha_x\rangle$, the eavesdropper Eve obtains the attenuated coherent state $|\sqrt{1-\eta} \, \alpha_x\rangle$. 
According to our model, this entails manipulations of coherent states, which are a particular family of pure, Gaussian states. 
Note that we cannot directly apply the results of Ref.~\cite{WildeGaussian} concerning Gaussian states, because, as we will see below, in communication protocols with PSK encoding one deals with non-Gaussian states arising from the sum of Gaussian states.
\paragraph{Thermal noisy channel---}
This model generalizes the previous description to a scenario where the communication is also affected by excess noise. Similar to the pure-loss case, the channel model includes a beam splitter with transmittance $\eta$. However, while in the pure-loss setting Eve injects the vacuum state into the free input branch of the beam splitter, here she owns a two-mode squeezed vacuum (TMSV) state, characterized by a mean photon number $\mu$. Eve injects one mode of the TMSV, which is in a thermal state, into the beam splitter to interact with Alice's signal, while the other mode is kept by her and evolves trivially. This setup allows Eve to purify the state she uses to introduce noise in the channel. In particular we characterize this channel through the \textit{excess noise} parameter $\xi$; once the average photon number $\mu$ and the transmittance $\eta$ are fixed, the \textit{excess noise} is defined as~\cite{Lutk2019}
\begin{align}
    \xi = \frac{\eta\mu}{2(1-\eta)} \, . 
\end{align}
Further mathematical details and explicit derivations of the states are found in Appendix \ref{sec:thermal}.
We have fully defined our optimization problem: for any given $\eta\in[0,1]$, since the attack model $\rho_{YE}$ is known, our task reduces to maximize the lower bounds on the secret-key rate in Eqs.~\eqref{RateSand}, \eqref{RateAEP} and \eqref{RateBound}, for BPSK and QPSK protocols, over the entropic parameter $a$ and the coherent state amplitude $|\alpha|$. This full optimization is performed specifically for the pure-loss case; for the thermal noisy channel, due to the high computational demand required to evaluate the entropic quantities numerically, we do not perform an optimization over the protocol parameters.

\section{Pure loss scenario: computation of conditional entropies for $N$--PSK protocols}
\label{sec:CondEnt}

Before proceeding with the optimization of the lower bounds on the secret-key rate, we present explicit expression for the conditional quantum entropies of a CQ state of the form of Eq.~(\ref{CQdef}).
The details of the derivations are in the Appendix~\ref{renyiproof}. 
First, the Petz-Rényi entropy reads
\begin{align}
    H_a^\downarrow (Y|E) 
 =
    \frac{1}{1-a}
    \log \left(
    \sum_{y=0}^{N-1}  \frac{1}{N^a} \,
    \mathrm{tr}
    \left( 
     \rho_{E|y}^a \;\rho_E^{1-a} 
    \right )
    \right) \, ,
    \label{petzl4}
\end{align}
where for BPSK and QPSK, and for the channel models considered here, the random variable $Y$ is uniform with $p_y = 1/N$. 
Furthermore, for our channel models the following relation holds
\begin{align}
    \rho_{E|y \oplus t} = U_t \, \rho_{E|y} \, U_t^\dag \, ,
    \label{symmetry00}
\end{align}
where $\{ U_t \}_{t=0, \dots, N-1}$ is a finite symmetry group of unitary matrices, for $t=0, \dots, N-1$, and $\oplus$ denotes summation modulo $N$.
This allows a further simplification of Eq.~\eqref{petzl4}:
\begin{equation}\label{HagiuMAIN}
    H_a^\downarrow (Y|E) = \log N + \frac{1}{1-a}
    \log \left\{
    \mathrm{tr}
    \left( 
     \rho_{E|0}^a \;\rho_E^{1-a} 
    \right )
    \right\} \,.
\end{equation}

For the optimized Petz-Rényi entropy we obtain
\begin{align}
       H_a^\uparrow(Y|E) 
     = -\frac{a}{1-a}\left(\log N-
    \log
    \left[
    \mathrm{tr}
    \left(
    \sum_{y=0}^{N-1}   \,
    \rho_{E|y}^a
    \right)^\frac{1}{a}
    \right] 
    \right)\, .
    \label{HasuMAIN}
\end{align}

Applying the symmetry argument as above to the sandwiched Rényi entropy we obtain
\begin{align}\label{SandgiuMAIN}
    \tilde{H}_a^\downarrow (Y|E) 
    = \log N+\frac{1}{1-a} \,
    \log{
    \mathrm{tr} \left[
\left(\rho_{E}^{\frac{1-a}{2 a}}
\rho_{E|0}\;
\rho_{E}^{\frac{1-a}{2a}}
\right)^a
    \right]
    } \, .
\end{align}

Finally we consider the optimized sandwiched Rényi entropy. {As detailed in the Appendix~\ref{renyiproof}, we obtain}
\begin{align}
\tilde{H}_a^\uparrow (Y|E)
& \geq \log N  \,\nonumber\\
+ \frac{1}{1-a} 
   & \log
    \left(
     \inf_{\sigma \in \mathfrak{I}} 
    \mathrm{tr} \left[
\left(\sigma_{E}^{\frac{1-a}{2 a}}
\rho_{E|0}\;
\sigma_{E}^{\frac{1-a}{2a}}
\right)^a
    \right]
    \right) 
    \label{HsupMAIN} \, ,
\end{align}
where $\mathfrak{I}$ denotes the set of invariant states under the action of a symmetry group $\{ U_t \}$, \textit{i.e.}
$\sigma_E = U_t \, \sigma_E \, U_t^\dag \, $.

\subsection{Conditional entropies for BPSK}\label{subsec:BPSK}
We are now ready to specialize the different computations of the entropic functionals that we showed at the beginning of this section to $N=2$ and $N=4$, \textit{i.e.}~for BPSK and QPSK protocols. For now we assume a pure-loss channel.

In BPSK, Alice prepares, with equal probability, one of the two coherent states, $| \alpha_x \rangle =
| (-1)^x \alpha \rangle$, 
for $x=0,1$, where $\alpha = |\alpha|$.
In turn, Bob measures the received states by homodyne detection in one of the two quadrature operators, either $q$ or $p$.
Measurements of the $p$ quadrature are used for channel estimation.

We consider QKD in \textit{reverse reconciliation}, \textit{i.e.}~the key is determined by the discretized output $y$ of Bob's measurement. In detail, for key extraction, the output of the measurements of the $q$ quadrature are used; a phase-space representation of this scheme is shown in Fig.~\ref{bpskspace}.
{Assuming a pure-loss channel from Alice to Bob, with transmittance $\eta$ (see Section~\ref{sec:model}), then} the probability density on a point of the phase space, for Bob measuring $q$ upon receiving  $|(-1)^x\sqrt{\eta} \, \alpha\rangle$ is
\begin{align}
    p\left( q|\, (-1)^x \sqrt{\eta} \,\alpha \,\right) = \frac{1}{\sqrt{\pi}} \, e^{- \left(
    q- (-1)^x\sqrt{2\eta} \, \alpha
    \right)^2} \, .
    \label{pyxBPSK}
\end{align}
Bob's outcomes are discretized into a binary value $y=0,1$, where $y=0$ if the measured value of the quadratures is $q>0$, and $y=1$ otherwise.  Therefore the conditional probabilities  $p\left( y |x \right)$ are  
\begin{align}
    p(y=x|x) &= 
    \frac{1}{\pi}
    \int_0^{+\infty}dq \;
    e^{- \left(
    q- \sqrt{2\eta} \, \alpha
    \right)^2} \nonumber\\
    &= 
    \frac{1}{2} \left( 
    1 + \text{erf}[\sqrt{2\eta} \, \alpha]
    \right)\,, 
    \label{py=x}\\
     p(y\neq x|x) &= 
    \frac{1}{\pi}
    \int_{-\infty}^{0}dq \;
    e^{- \left(
    q- \sqrt{2\eta} \, \alpha
    \right)^2}\nonumber\\
    &= 
    \frac{1}{2} \left( 
    1 - \text{erf}[\sqrt{2\eta} \, \alpha]
    \right) 
    \, .
    \label{pydifx}
\end{align}
Now that we computed~\eqref{py=x} and~\eqref{pydifx}, the object of interest for the reverse reconciliation scheme is exactly Eq.~\eqref{CQdef} which is specialized for BPSK as follows
\begin{align}
    \rho_{YE}
    & = \frac{1}{2} 
\sum_{y=0,1} |y\rangle_Y \langle y| \otimes
\rho_{E|y} \label{aaa}
\\
& = \frac{1}{2} 
\sum_{y=0,1} |y\rangle_Y \langle y| \otimes
\sum_{x=0,1} p(x|y) \, \rho_{E|x}
\nonumber\\
& = \frac{1}{2} 
\sum_{y=0,1} |y\rangle_Y \langle y| \otimes
\sum_{x=0,1} p(x|y)  \,
    | (-1)^x \gamma\rangle_E 
    \langle (-1)^x \gamma|
    \, ,\nonumber
\end{align}
where $\gamma = \sqrt{1-\eta}\;\alpha$, and we observe that $p(y|x) = p(x|y)$.
In particular:
\begin{align}
\rho_{E|y=0}
 = 
    \frac{1}{2} \sum_{x=0,1}&\left( 
    1 + (-1)^x\text{erf}[\sqrt{2\eta} \, \alpha ]
    \right)
    |(-1)^x \gamma\rangle 
    \langle (-1)^x \gamma|
    \, , \\
\rho_{E|y=1}
 = 
    \frac{1}{2} \sum_{x=0,1}&\left( 
    1 + (-1)^{x+1}\text{erf}[\sqrt{2\eta} \, \alpha ]
    \right)
    \langle (-1)^x \gamma|
    \, .
\end{align}
Finally, the average state is
\begin{align}
    \rho_E = 
    \frac{1}{2}
    \left(
    |\gamma \rangle \langle \gamma| +
|-\gamma \rangle \langle -\gamma| \right) \, .  
\end{align}

It is convenient to express the coherent states in the orthogonal basis $\{ \ket{\psi_+} , \ket{\psi_-} \}$
\begin{align}
    \ket{\psi_\pm}  = \frac{|\gamma\rangle \pm |-\gamma\rangle}{\sqrt{2 c_\pm}} \, , 
\end{align}
such that
\begin{align}
    |\pm\gamma\rangle  = 
    \frac{ \sqrt{
    c_+} |\psi_+\rangle 
    \pm \sqrt{c_-} |\psi_-\rangle }{\sqrt{2}} \, .
\end{align}
The coefficients $c_\pm$ are for normalization. Using the formula for the scalar product of two coherent states, $\langle \beta_1 | \beta_2 \rangle = e^{-\frac{1}{2}\left( 
|\beta_1|^2 + |\beta_2|^2 - 2 \beta_1^* \beta_2
\right)}$, we obtain
\begin{align}
c_\pm = 1 \pm e^{-2 |\gamma|^2} \, .
\end{align}
From this we get
\begin{align}
    |\pm\gamma\rangle \langle \pm \gamma| 
    = \frac{1}{2}
    \big(&
    c_+ |\psi_+\rangle\langle\psi_+|
    + c_- |\psi_-\rangle\langle\psi_-|\nonumber\\
    &
    \pm \sqrt{c_+ c_-} \left( |\psi_+\rangle\langle\psi_-| + \mathrm{hc}
    \right)\,
    \big).
\end{align}
Therefore,
\begin{align}
    \rho_E 
    & = 
    \frac{1}{2}
    \left(
    c_+ |\psi_+\rangle\langle\psi_+|
    + c_- |\psi_-\rangle\langle\psi_-|
    \right) \, , \\
\rho_{E|y=0}
=&     \frac{1}{2}
    \biggl(
    c_+ |\psi_+\rangle\langle\psi_+|
    + c_- |\psi_-\rangle\langle\psi_-|\nonumber\\
    &+ \text{erf}[\sqrt{2\eta}|\alpha|]
    \sqrt{c_+ c_-}
     \left( |\psi_+\rangle\langle\psi_-| + \mathrm{hc}
    \right)    
    \biggr) \, , \\
\rho_{E|y=1}
& =     \frac{1}{2}
    \biggl(
    c_+ |\psi_+\rangle\langle\psi_+|
    + c_- |\psi_-\rangle\langle\psi_-|\nonumber\\
    &- \text{erf}[\sqrt{2\eta}|\alpha|]
    \sqrt{c_+ c_-}
     \left( |\psi_+\rangle\langle\psi_-| + \mathrm{hc}
    \right)    
    \biggr) \, , 
    \label{rhoey1}
\end{align}
where in the last two lines we used the relation
\begin{align}
     \frac{
    |\gamma\rangle \langle \gamma| 
        -     |-\gamma\rangle \langle -\gamma|
        }{2}
    & = \frac{\sqrt{c_+ c_-}}{2}
     \bigg( |\psi_+\rangle\langle\psi_-| + \mathrm{hc}
    \bigg)
    \, .
\end{align}

For these particular states we can find an analytical expression for some of the quantum conditional entropies. 

The first result that we obtain in this paper is that specializing Eq.~\eqref{CQdef} with Eqs.~\eqref{aaa}-\eqref{rhoey1} for BPSK the following conditional entropies admit an analytical expression. Let us define
\begin{gather}
\kappa = e^{2\alpha^2(\eta-1)}\, ,\,
r = \operatorname{erf}\!\big(\sqrt{2\eta}\,\alpha\big),\,
g = \sqrt{\,1 + (\kappa^{-2}-1)\,r^2\,} \, , \\
\theta = \operatorname{artanh}(\kappa g) \, ,\quad
\phi = \operatorname{artanh}(\kappa) \, ,
\quad \sech{}\, x = 1/\cosh x \, .
\end{gather}
We obtain:
\begin{widetext}
\begin{itemize}
    \item Petz-Rényi entropy:
\begin{align}
\label{eq:Ha-elegant}
H_a^\downarrow(Y|E)
=& 1 + \frac{1}{1-a}\,
\log\!\Big\{
\sech{}^a\;\!\theta\;
\sech{}^{\,1-a}\;\!\phi\;
\big[
\cosh(a\theta)\cosh((1-a)\phi)
+g^{-1}\sinh(a\theta)\sinh((1-a)\phi)
\big]
\Big\} \, ;
\end{align}
    \item Optimized Petz Rényi entropy:
\begin{equation}
H_a^\uparrow(Y|E)
= \frac{a}{1-a}\,
\log\!\Bigg[
2^{\,\frac{1}{a}-2}\,\sech\theta\,
\Big(
\big(\cosh(a\theta)+g^{-1}\sinh(a\theta)\big)^{\!1/a}
+
\big(\cosh(a\theta)-g^{-1}\sinh(a\theta)\big)^{\!1/a}
\Big)
\Bigg]\, ;
\label{PetzUpAn}
\end{equation}

\item  Sandiwiched Rényi entropy:
\begin{equation}
\label{eq:Htilde-down-elegant}
\tilde{H}_a^\downarrow(Y|E)
= -\frac{2a}{1-a}
\;+\;
\frac{1}{1-a}\,
\log\!\Bigg[
2^{\,a}\,\sech{}\,\phi\;
\Big(
\big(\cosh(\tfrac{\phi}{a})-\Delta\big)^{a}
+
\big(\cosh(\tfrac{\phi}{a})+\Delta\big)^{a}
\Big)
\Bigg],
\qquad \Delta := \sqrt{\,r^{2} + \sinh^{2}\!\big(\tfrac{\phi}{a}\big)\,}.
\end{equation}
\end{itemize}
\end{widetext}
The derivation of the analytical expression \eqref{PetzUpAn} for the optimized Petz-Rényi entropy $ H_a^\uparrow(Y|E)$ was possible thanks to Lemma 1 of Ref.~\cite{BertaJMP}, which provides an explicit expression of the optimal form for this quantity. 
\begin{figure}
    \centering
    \subfloat[BPSK protocol.]{
        \includegraphics[width=\linewidth]{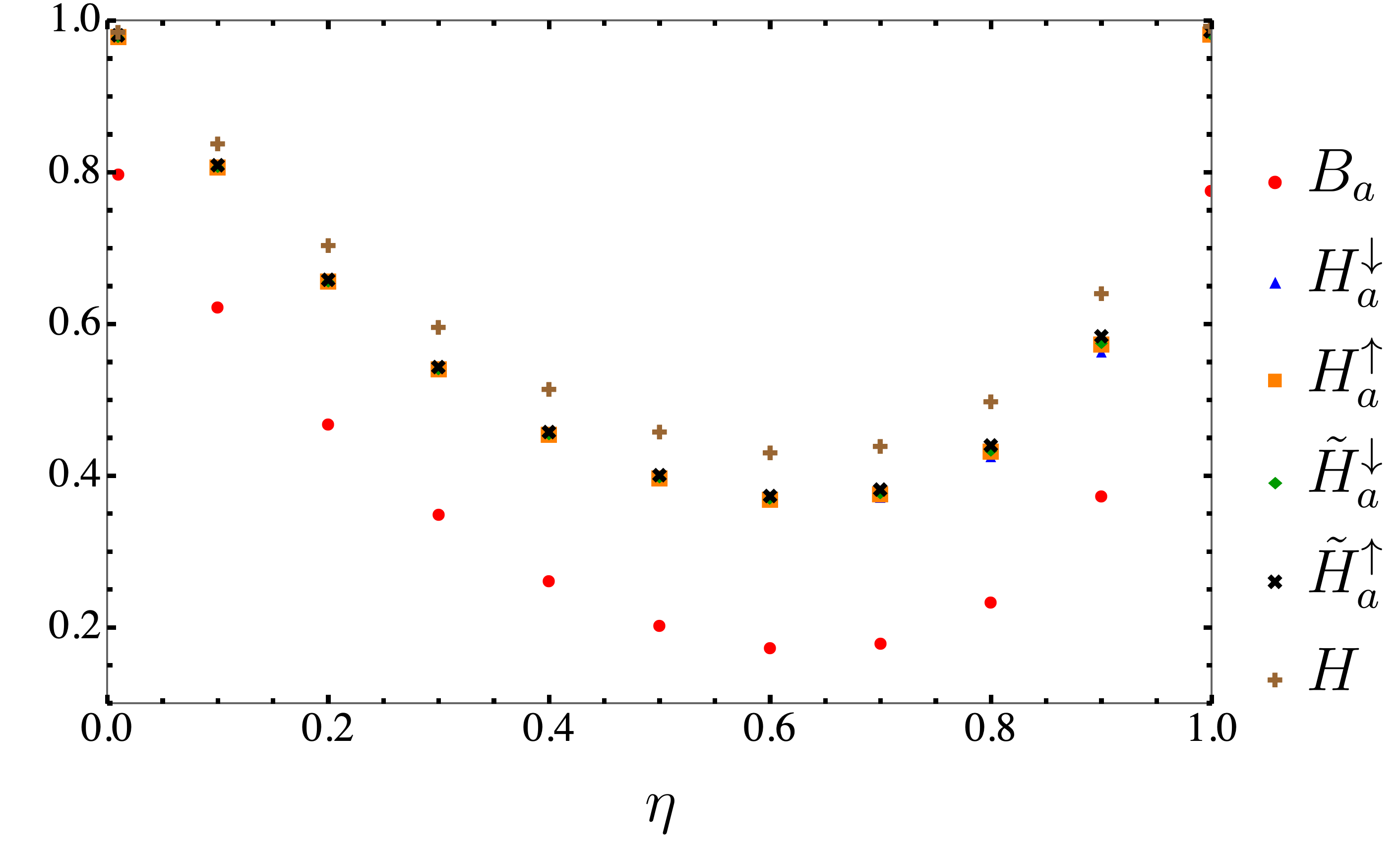}
        \label{BPSKcomp}
    }\\
    \subfloat[QPSK protocol.]{
        \includegraphics[width=\linewidth]{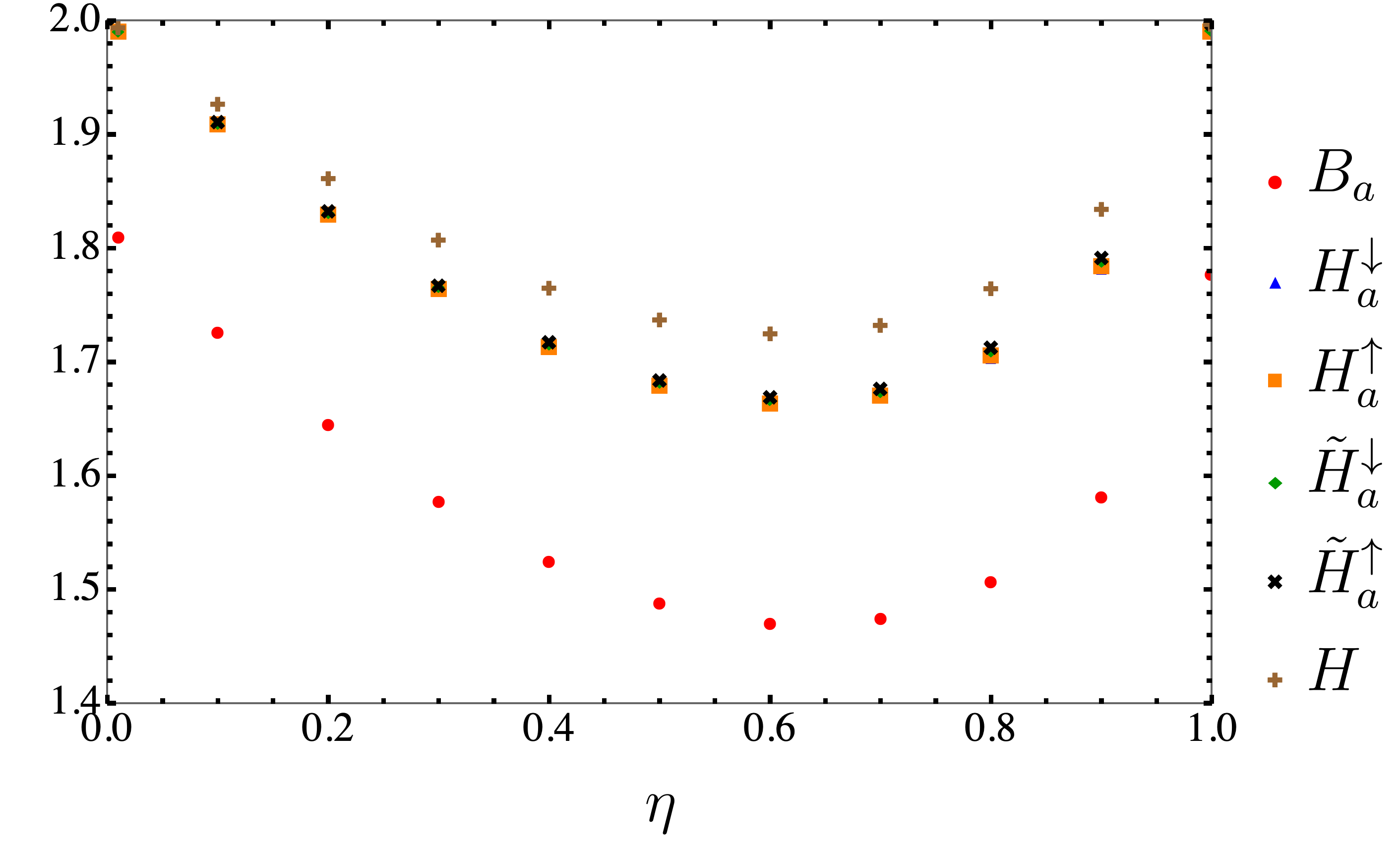}
        \label{QPSKcomp}
    }
    \caption{\justifying  Comparison between different kinds of entropies, computed on the CQ states at $\alpha=1$ and $a=1.2$ on $\rho_{YE}$ in Eqs.~(\ref{aaa}) and (\ref{CQ_qpsk}) respectively for BPSK in \ref{BPSKcomp} and in \ref{QPSKcomp} for QPSK protocol.
    In particular: Petz-Rényi entropy $H^\downarrow_a(Y|E)$ (\ref{HagiuMAIN}); optimized Petz-Rényi entropy $H^\uparrow_a(Y|E)$ (\ref{HasuMAIN}); sandwiched Rényi entropy $\tilde H^\downarrow_a(Y|E)$ (\ref{SandgiuMAIN}); optimized sandwiched Rényi entropy $\tilde H^\uparrow_a(Y|E)$ (\ref{HsupMAIN}), von Neumann entropy $H(Y|E)$ (\ref{VonNeumann}); continuity bound $B_a(Y|E)$ (\ref{VKBound}). All the quantities converge to $1$ as $\eta\to0,1$ for BPSK and to $2$ for QPSK, while the minima are reached around $\eta_{\text{min}}\sim0.6$ in both cases.
    }
    \label{fig:overall}
\end{figure}

In Fig.~\ref{BPSKcomp} we show the different entropic functionals for fixed values of the parameters $|\alpha|=1$ and $a=1.2$, against the transmittance $\eta$.
First of all, we observe that all the different quantities converge to the maximum value $1$ as $\eta\to0,1$. Let us analyze these two particular cases. If $\eta\to1$, the quantum channel that transmits signals from Alice to Bob can be considered ideal, as no losses are ever experienced and Bob receives exactly what Alice sends. Then Eve's final state will not contain any information about Bob's measurements outcomes and her uncertainty about it will be maximal. On the contrary when $\eta\to0$ Eve will own Alice's signal unaltered. This implies that the eavesdropper can easily infer the bit of information encoded in the light signal after performing homodyne detection. This might seem the worst case-scenario, but actually Eve is interested in what Bob's measurement outcome will be, since this is what will end up in the final key string in reverse reconciliation. 
If $\eta\to0$ Bob receives a state which very close to the vacuum, then practically his homodyne detection will return a random variable of zero mean uncorrelated with Eve. 
Hence, again, Eve's uncertainty will reach the maximum value. 
The second thing that we can notice is that the different curves related to the different Rényi entropies follow the monotonicity relations discussed in Appendix~\ref{sec2}, as we could expect, with the von Neumann entropy $H(Y|E)$ as the upper bound for each value of $\eta$. We see that the difference between the various entropies is actually very small and becomes negligible when $a\to1$ since all these quantities converge to the von Neumann entropy. Finally, if we focus on $B_a(Y|E)$ in Eq.~(\ref{VKBound}), we see that it represents a lower bound not only for the sandwiched Rényi entropies, but for the Petz-Rényi entropies as well; the fact that this is the only quantity that does not tend to $1$ at the boundary of the transmittance range depends on the fact that the subtractive terms rigidly shift the von Neumann entropy down.

\subsection{Conditional entropies for QPSK}\label{subsec:QPSK}

We extend from two to four-state, QPSK protocol. 
Alice prepares, with equal probability, one of the four coherent states $| \alpha_k \rangle = | 
i^k e^{i \pi/4} \alpha \rangle$, for $k=0,1,2,3$ centered in a quadrant of the phase space as shown in Fig.~\ref{QpskScheme}.
Bob decoding is the heterodyne detection, yielding a continuous, complex-valued amplitude $\beta$. 
To this continuous output, Bob associates a discrete value $y=0,1,2,3$ that encodes two bits ($00$, $01$, $10$, $11$), corresponding to the four colored regions associated to $\ket{\alpha_y}$ in Fig.~\ref{QpskScheme}. For example, if $\beta$ belongs to the first quadrant in phase space, \textit{i.e.}~if $\mathsf{Re}( \beta), \mathsf{Im}( \beta) >0$, then $y=0$, and so on and so forth.
Bob measurement, combined with this discretization map, defines a random variable $Y$ with alphabet $\{0,1,2,3 \}$. 

For QPSK, we compute explicitly the probability for Bob to measure the output $y$ in correspondence with a preparation $k$
\begin{align}\label{pyk}
   p(y|k) = 
   \mathrm{tr} 
   \bigg[
   \left(
   \ket{\phi_k}_A\bra{\phi_k}\otimes \Lambda_y
   \right)
   \rho_{AB}
   \bigg
   ].
\end{align}
This case differs from BPSK, because we need to evaluate explicitly $\rho_{AB}$, Alice and Bob's state after the transmission and before Bob's measurement, which according to \eqref{rhoABstate} is given by
\begin{align}\label{rhoAB}
    \rho_{AB}= 
    \frac{1}{4}\sum_{x,x'=0}^3
    \ket{\phi_x}_A\bra{\phi_{x'}} \otimes 
\ket{\sqrt{\eta}~\alpha_x}_B\bra{\sqrt{\eta}~\alpha_{x'}} \nonumber\\
e^{-\alpha^2(1-\eta)
\left( 
1 - e^{-i(x'-x)\frac{\pi}{2}} 
\right)} \, .
\end{align}

The analogue of Eqs.~\eqref{py=x}-\eqref{pydifx} of BPSK, for QPSK is obtained using Eq.~\eqref{pyk} and Eq.~\eqref{eq:Lambda}:
\begin{align}
    p(y|k) = \begin{cases}
        \frac14
        \left[
        1+\text{erf}\left(
        \sqrt{\frac{\eta}{2}} \, \alpha 
        \right)
        \right]^2, \qquad y=k\\
         \frac14
         \left[
         1-\text{erf}
         \left( 
         \sqrt{\frac{\eta}{2}} \, \alpha
         \right)
         \right]^2, \qquad  |y-k| =2\\
         \frac14
         \left[
         1-\text{erf}^2
         \left( \sqrt{\frac{\eta}{2}} \, \alpha
         \right)
         \right], \qquad  |y-k|=1,3
    \end{cases}
    \label{pkj}
\end{align}
We can introduce the short-hand notations 
\begin{align}
    P_\pm=
    \frac12
    \left[
    1\pm\text{erf}
    \left(
    \sqrt{\frac{\eta}{2} \, }\alpha
    \right)
    \right]
    \label{eq:p_pm}
\end{align}
and gather them as
\begin{align}
     p(y|k)=p_{y,k}=
\left(\begin{array}{cccc}
P_{+}^{2} & P_{+}P_{-} & P_{-}^{2} & P_{+}P_{-}\\
P_{+}P_{-} & P_{+}^{2} & P_{+}P_{-} & P_{-}^{2}\\
P_{-}^{2} & P_{+}P_{-} & P_{+}^{2} & P_{+}P_{-}\\
P_{+}P_{-} & P_{-}^{2} & P_{+}P_{-} & P_{+}^{2}
\end{array}\right) \, .
\label{matrix}
\end{align}
and Eq.~\eqref{CQdef} for QPSK specializes as
\begin{align}
    \rho_{YE}
    & = \frac{1}{4} 
\sum_{y=0}^3 |y\rangle_Y \langle y| \otimes
\rho_{E|y}
\label{CQ_qpsk}\nonumber
\\
& = \frac{1}{4} 
\sum_{y=0}^3 |y\rangle_Y \langle y| \otimes
\sum_{k=0}^3 p(k|y) \, \rho_{E|k}\nonumber
\\
& = \frac{1}{4} 
\sum_{y=0}^3 |y\rangle_Y \langle y| \otimes
\sum_{k=0}^3 p(k|y)  \,
    |\gamma_k \rangle_E 
    \langle  \gamma_k|
    \, ,
\end{align}
where $
    \ket{\gamma_k}_E = \ket{
    i^k
    e^{i\frac{\pi}{4}}
    \sqrt{1-\eta}\; \alpha
    }_E
    $ and $p(y|k)=p(k|y)$.
In particular
\begin{align}
\rho_{E|y}
 = \sum_{k=0}^3 p(k|y)  \, 
    | \gamma_k\rangle _E
    \langle  \gamma_k|
    \, ,
\end{align}
is computed using $p(y|k)$ from the matrix \eqref{matrix}. 
Finally, the average state is
\begin{align}
    \rho_E = 
    \frac14
    \sum_{y=0}^3
    \rho_{E|y}=
    \frac14
    \sum_{y=0}^3
    \ket{\gamma_y}_E\bra{\gamma_y}
    \, .
\end{align}
As in the case of BPSK, we can express the coherent states in an orthonormal basis, 
\begin{align}
    \ket{\psi_s}_E = 
    \frac{N_s}{2}
    \sum_{k=0}^3
    e^{-i\frac{\pi}{2}sk}\ket{\gamma_k}_E, \qquad s=0,1,2,3 \, .
\end{align}
From this we obtain the inverse formulas
\begin{align}
    \ket{\gamma_k}_E = \sum_{s=0}^3 \frac{1}{2N_s} e^{i\frac{\pi}{2}sk} \ket{\psi_s} _E,\qquad k=0,1,2,3 \, .
\end{align}
Imposing the normalization condition $\braket{\psi_s|\psi_s}=1$ and recalling that $\gamma = \sqrt{1-\eta} \; \alpha $, we can find the expression of the normalization constants 
\begin{align}
    N_s^{-2}
     = 1
       +e^{-2\gamma^2}\cos(\pi s)
       +2e^{-\gamma^2}\cos\Bigl(\gamma^2+\tfrac{\pi}{2}s\Bigr) \, .
\end{align}
From this, we obtain
\begin{align}
   \rho_{E|k}= 
   \ket{\gamma_k}_E\bra{\gamma_k} 
   &=\sum_{s,s'=0}^3 
   \frac{1}{4N_s N_{s'}}
   e^{i\frac{\pi}{2}(s - s')k}
   \ket{\psi_s}_E\bra{\psi_{s'}}
   \, ,
\end{align}
and the conditional states are
\begin{align}
    \rho_{E|y} = 
    \sum_{s=0}^{3} \sum_{s'=0}^{3} 
    \left(
    \frac{1}{4 N_s N_{s'}} 
    \sum_{k=0}^{3} 
    p(y|k) \, e^{i \frac{\pi}{2}(s - s')k} 
    \right)
    |\psi_s\rangle \langle \psi_{s'}| \, .
\end{align}
We note that Eve's average states is diagonal in the orthogonal basis, \textit{i.e.}
\begin{align}\label{rhoEQPSK}
    \rho_E = \frac14
    \sum_{s=0}^3 
    \frac{1}{N_s^2} 
    \ket{\psi_s}\bra{\psi_s} \, .
\end{align}
The analytical expressions from Eq.~\eqref{rhoAB} to Eq.~\eqref{rhoEQPSK} simplify the final numerical evaluation of $H_a^\downarrow$, $H_a^\uparrow$, $\tilde H_a^\downarrow$, and others shown in Fig.~\ref{QPSKcomp}. However, obtaining a fully analytical derivation analogous to the BPSK case in Eqs.~\eqref{eq:Ha-elegant}, \eqref{PetzUpAn}, and \eqref{eq:Htilde-down-elegant} remains an open challenge.

In Fig.~\ref{QPSKcomp} there is a comparison between different entropic functionals, for fixed values of $\alpha$ and $a$, against the transmittance $\eta$. The trends of these quantities, as well as the motivation behind them, are identical to those of the BPSK protocol (the only significant difference is the maximum value that these entropies assume as $\eta\to0,1$, which is $2$). Also in this case we found that the minimum value for all the quantities is reached around 
$\eta_{\text{min}}\sim0.6$. 

Now that we verified the monotonicity relations \eqref{monotone} with the conditional entropies at fixed $a$ evaluated on the states at coherent amplitude depending on $|\alpha|$ and the transmittance $\eta$, the final goal is to estimate the secret-key rate for BPSK and QPSK optimizing over $a$ and $\alpha$ at a given value of trasmittance $\eta$ and \textit{finite} rounds $n$.
\begin{figure}
    \centering
    \subfloat[BPSK protocol.]{
        \includegraphics[width=\linewidth]{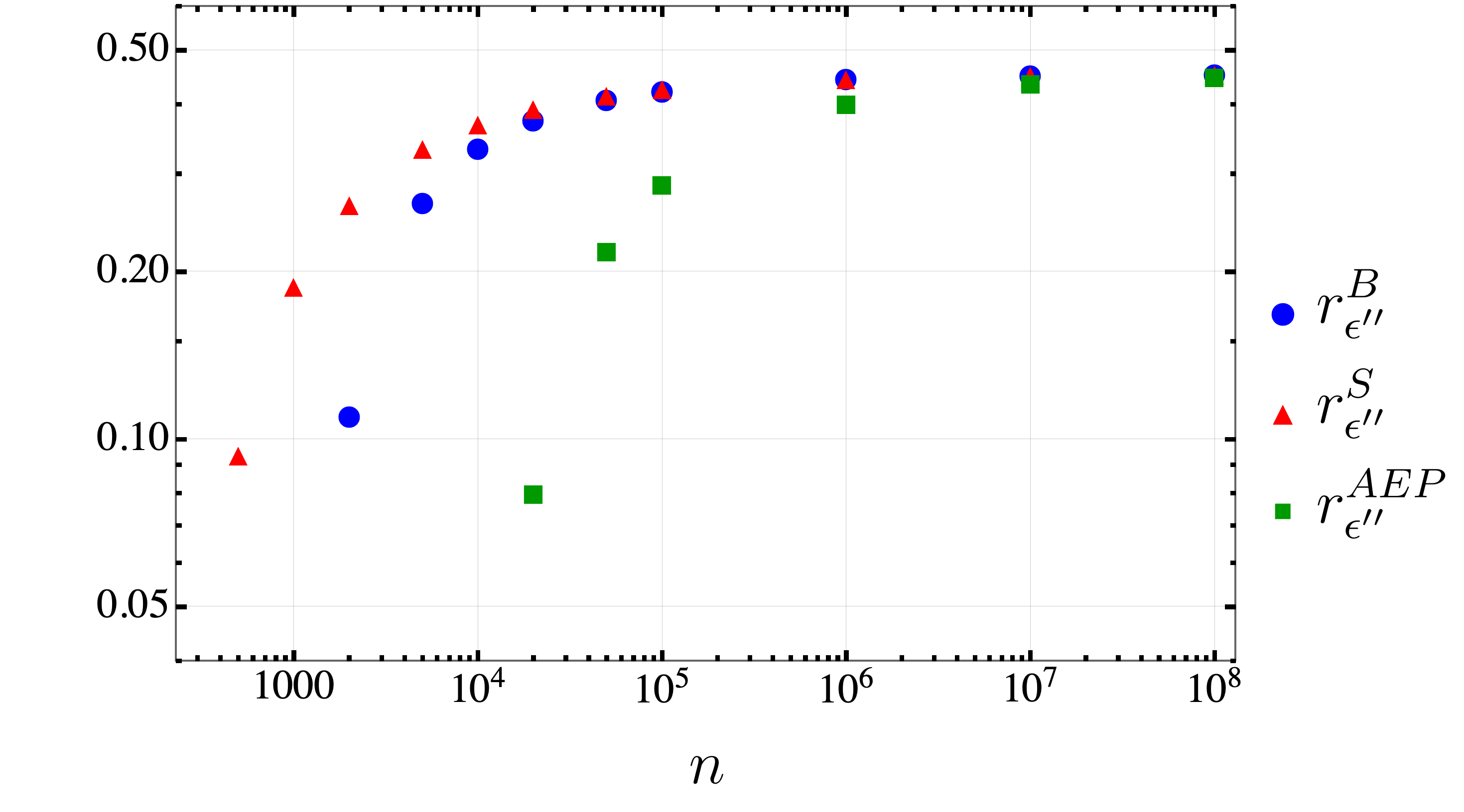}
        \label{rateBPSK}
    }\\
    \subfloat[QPSK protocol.]{
        \includegraphics[width=\linewidth]{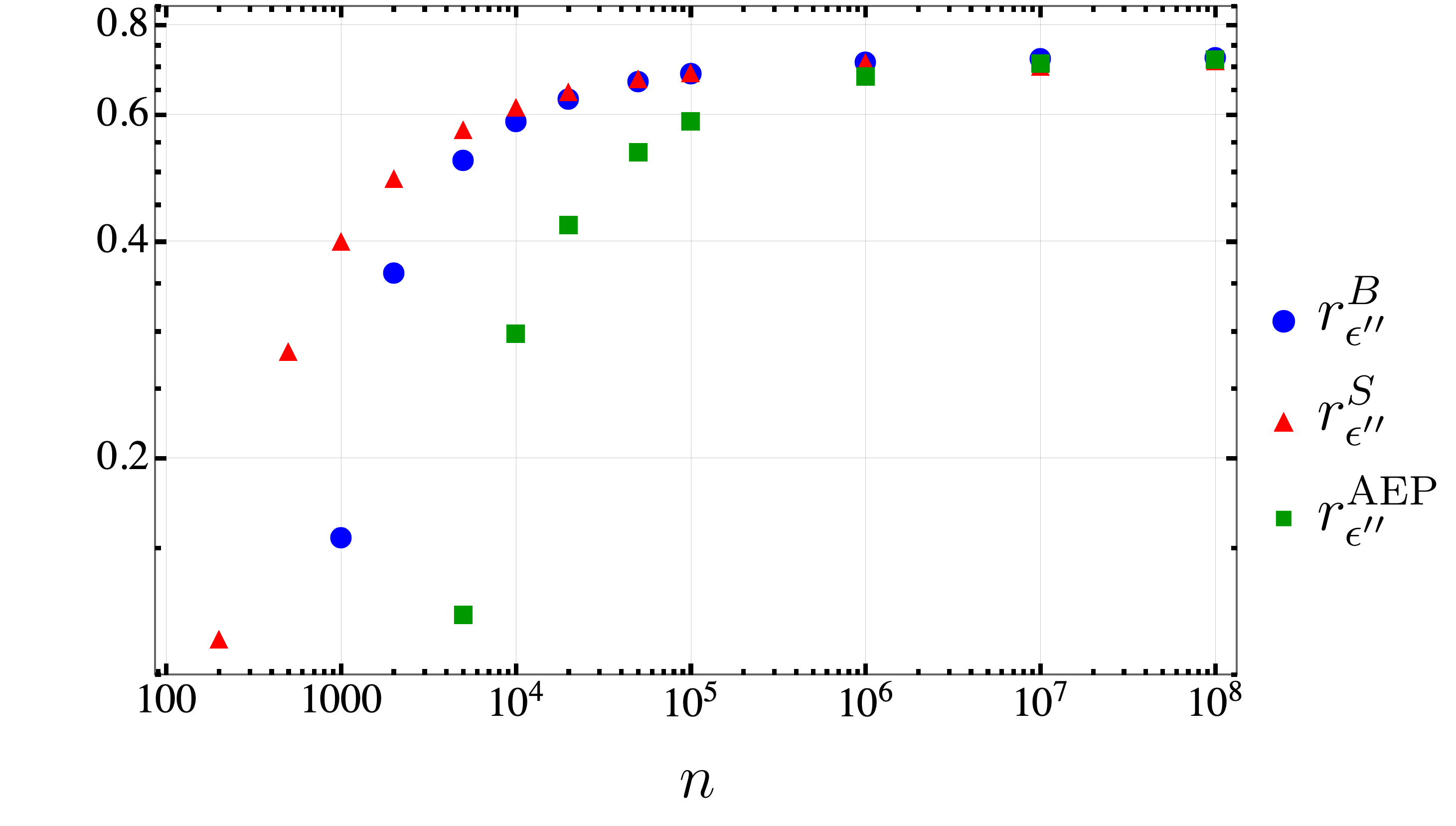}
        \label{rateQ}
    }
    \caption{\justifying Secret-key rate estimates versus the number of signals $n$, for a fixed value of the transmittance at $\eta=0.9$, $\epsilon= \epsilon'=10^{-8}$. The plot compares the bounds $r_{\epsilon''}^{S}$, $r_{\epsilon''}^{B}$, $r_{\epsilon''}^{AEP}$ respectively in  (\ref{ratesandBPSK}), (\ref{rateboundBPSK}) and   \ref{rateAEPBPSK}) with $\epsilon'' = \epsilon+\epsilon'$, for BPSK in \ref{rateBPSK} and for QPSK in \ref{rateQ}. The estimators $r_{\epsilon''}^{S}$ and $r_{\epsilon''}^{B}$  are optimized over the parameters $a$ and $\alpha$, while $r_{\epsilon''}^{AEP}$ is optimized over $\alpha$ alone. Asymptotically, the quantities converge to the same value, while in the finite-size region $r_{\epsilon''}^{S}$ is the tightest bound, being the only one still non-vanishing when $n<10^3$, and dropping to zero as $n < 5\cdot10^2$.
    }
    \label{fig:overallrate}
\end{figure}
\section{Pure loss scenario: estimates for the secret-key rate for N-PSK protocol}
\label{keyrateBPSK}

Assuming collective attacks performed by the eavesdropper on a $n$-round protocol, we compare different estimates for the key rates, obtained by bounding the number of secret bits $\ell_{\epsilon''}$ as in Eq.~(\ref{hash}) in three different ways.
We label the three bounds with the index $W$, where:
\textit{(i)} $W=\mathrm{S}$ for the bound obtained directly from the sandwiched Rényi entropy in Eq.~\eqref{RateSand};
\textit{(ii)} $W=\mathrm{AEP}$ for the bound obtained from the AEP in Eq.~\eqref{RateAEP}; 
and \textit{(iii)} $W=\mathrm{B}$ for the bound obtained from the second-order espansion in Eq.~\eqref{RateBound}.

We also need to consider the number $\ell_{leak}$ of bits that leak to the environment during the routine of error correction. In reverse reconciliation, Alice reconstructs Bob's key string starting from the string in her possession and Bob has to transmit a certain amount of error-correcting information. 
For simplicity we estimate $\ell_{leak}$ using the asymptotic, information-theoretic bound given by the conditional Shannon entropy.
Let us denote as $X$ the uniformly distributed random variable representing Alice's choice for state preparation. Then information theory quantifies the asymptotic error-correction information needed by Alice in terms of the conditional Shannon entropy \cite{SlepianWolf}
\begin{align}
    H_N(Y|X)=&\sum_{k=0}^{N-1}p(k)H(Y|X=k)\nonumber\\
    =&-\sum_{k=0}^{N-1}p(k)
   \sum_{y=0}^{N-1}p(y|k)
    \log\,
    p(y|k)
    \, .
\end{align}
Hence, if we consider $n$ rounds of the protocol, the total amount of information leakage will be $\ell_{leak} = nH_N(Y|X)$.

We compare the following estimates of the secret-key rates, according to (\ref{hash}), with $\epsilon'' = \epsilon+\epsilon'$,
$r_{\epsilon''}^W = (\ell_{\epsilon''}^W -\ell_{leak})/n\blk$, for $W$ that distinguish the quantities $\mathrm{W}=\mathrm{S},\mathrm{AEP},\mathrm{B}$.
\begin{enumerate}
    \item From the optimized Sandwiched Rényi $a$-entropy $\tilde{H}_a^\uparrow(Y|E)$, ($a>1$) (\ref{RateSand}),
    \begin{align}
        r_{\epsilon''}^S \ge
        \tilde{H}_a^\uparrow(Y|E)
        +\frac{ 1+ 2 \log{\epsilon'}}{n}  \,  
        -\frac{g(\epsilon)}{n(a-1)}
        -H_N(Y|X);
        \label{ratesandBPSK}
    \end{align}
    \item From the AEP (\ref{RateAEP}),
    \begin{align}
        r_{\epsilon''}^{AEP}\ge
        H(Y|E) 
         +\frac{ 1+ 2 \log{\epsilon'}}{n} \,
        - \frac{\delta(\epsilon)}{\sqrt{n}} 
        -H_N(Y|X)\;;
         \label{rateAEPBPSK}
    \end{align}
    \item From the von Neumann entropy continuity bound ($a\in(1,2]$) (\ref{RateBound}),
    \begin{align}
        r_{\epsilon''}^{B} \ge B_a(Y|E) 
         +\frac{ 1+ 2 \log{\epsilon'}}{n} \,
        -\frac{g(\epsilon)}{n(a-1)}
        -H_N(Y|X)\; .
         \label{rateboundBPSK}
    \end{align}
\end{enumerate}
\subsection{Estimates for the secret-key rate for BPSK protocol}
For the BPSK case ($N=2$), we have
\begin{align}
    \frac{\ell_{leak}}{n} &= H_2(Y|X) = \sum_{k=0,1}p(k)H(Y|X=k)\nonumber\\
    &=-\sum_{k=0,1}p(k)
    \sum_{y=0,1}p(y|k)
    \log\,
    \left[
    p(y|k)
    \right]\nonumber\\
    &=- \sum_{y=0,1}p(y|0)
    \log\,
    \left[
    p(y|0)
    \right]
    =- \sum_{y=+,-}p_y
    \log\,
    (p_y).
\end{align}
with $p_\pm =\frac12[1\pm \text{erf}(\sqrt{2\eta} \, \alpha)]$.

We show in Fig.~\ref{rateBPSK} the three different estimations of the key rate, computed for a fixed value of $\eta=0.9$, plotted against $n$; 
here we set $\epsilon = \epsilon'= 10^{-8}$, and thus 
$\epsilon'' = 2 \times 10^{-8}$. 
In particular the rates $r_{\epsilon''}^S$ and $r_{\epsilon''}^B$ are optimized over the parameters $a$ and $\alpha$, while $r_{\epsilon''}^{AEP}$ is optimized only over $\alpha$, not being depending on the Rényi parameter $a$. We can see that asymptotically the three quantities reach all the same value which is around $r\sim0.45$; in fact in this regime all the finite-size correction vanish and the optimization brings $a\to1$, both for $\tilde{H}_a^\uparrow(Y|E)$ and $B_a(Y|E)$ which clearly tend to $H(Y|E)$. When $n$ gets smaller, \textit{i.e.}~around $n\sim10^6$, $ r_{\epsilon''}^{AEP}$ is the first quantity dropping to zero because of the worse correction term $\sim1/\sqrt{n}$ and in particular we see that $ r_{\epsilon''}^{AEP} >0$ only for $n>10^4$. When the number of signals gets even smaller, \textit{i.e.}~$n<10^4$, we observe $r_{\epsilon''}^{S}>r_{\epsilon''}^{B}$ revealing the optimized sandwiched Rényi entropy as the best estimator of $\ell_{\epsilon''}$ in this regime. A very important fact about $r_{\epsilon''}^{S}$ is that it is the only non-vanishing estimation of the key rate when $n<10^3$, assuming still non-zero value $r_{\epsilon''}^{S}\sim 0.09$ for $n\sim500$. Finally we mention the fact that the optimization brings the amplitude  $\alpha\in [0.93,0.99]$ for $r_{\epsilon''}^{S}$,  $\alpha\in [0.89,0.95]$ for $r_{\epsilon''}^{B}$, while the optimal value for  $r_{\epsilon''}^{AEP}$ is always around $\alpha\sim0.95$. As for the Rényi parameter $a$, the optimal values get larger and larger as $n$ decreases, varying in the range $\alpha\in(1,1.25]$. Only in the case of $r_{\epsilon''}^{S}$, when $n$ becomes very small, the optimal value of $a$ rises sharply, indicating the transition to a regime where the min-entropy provides the optimal bound.
This transition is shown in Fig.~\ref{fig:trans}.

\begin{figure}
    \centering
    \includegraphics[width=\linewidth]{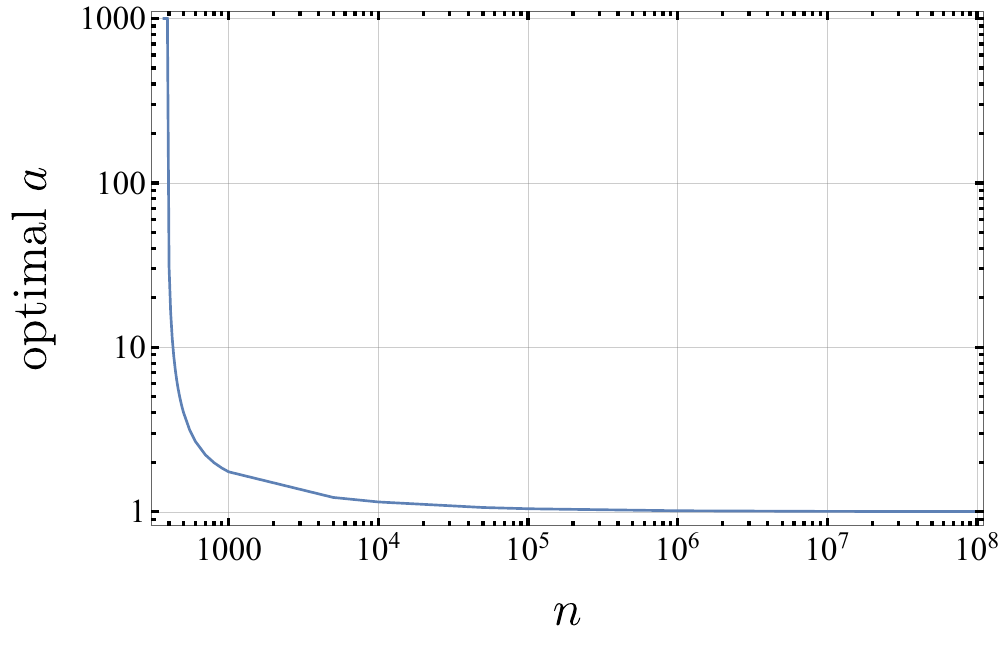}
    \caption{\justifying Optimal value of $a>1$ for the estimator $r_{\epsilon''}^S$ based on a direct calculation of the sandwiched Rényi entropy. See other numerical details in the \texttt{Mathematica} code available on this github folder~\cite{gioscaNPSK}.}
    \label{fig:trans}
\end{figure}

\subsection{Estimates for the secret-key rate for QPSK protocol}

The three different estimations of the secret-key rate in Eqs.~\eqref{ratesandBPSK}-\eqref{rateboundBPSK} for the QPSK protocol 
are computed on the CQ state \eqref{CQ_qpsk} and is used $H_{N=4}(Y|X)$ to estimate $\ell_{leak}$. Exploiting the conditional probability distribution in \eqref{pkj}-\eqref{matrix} we find
\begin{align}
     \frac{\ell_{leak}}{n} = H_4(Y|X) 
    &=-2\left[
    (P_+)\log (P_+) + (P_-)\log (P_-) 
    \right] \, .
\end{align}

In Fig.~\ref{rateQ} we show a comparison between  
$r_{\epsilon''}^{S}$, $r_{\epsilon''}^{B}$ and $r_{\epsilon''}^{AEP}$, plotted against the number of signals $n$; 
here again we set $\epsilon= \epsilon'=10^{-8}$. 
An optimization over the parameters $\alpha$ and $a$ was performed on these quantities, while keeping fixed the value of the transmittance at $\eta = 0.9$. The trends for the QPSK protocol are quite identical to the one seen for the BPSK protocol, both in the asymptotic regime and in the finite-size region. Some differences arise if we look at the optimal parameters, in particular in the optimal values of the amplitude $\alpha$; indeed for $r_{\epsilon''}^{S}$ the optimal values fall in the range $\alpha\in[1.58,1.67]$, for $r_{\epsilon''}^{B}$ in $\alpha\in[1.60,1.66]$, while for $r_{\epsilon''}^{AEP}$ the optimal amplitude is always around $\alpha \sim 1.66$. Comparing these values with the ones obtained in the BPSK case, we see that a higher amplitude for the coherent states is needed to achieve an optimal secret-key rate in the QPSK protocol.
Analogously to what happens in the BPSK protocol, the optimal values of the amplitude get larger asymptotically, while the converse is true for the optimal values of the Rényi parameter $a$, which get smaller for larger $n$, approaching $a\to1$ as $n$ increases. Again, just as we said in the BPSK case, the optimal values of $a$ rises sharply as $n$ becomes very small for $r_{\epsilon''}^{S}$, setting the transition to regime where the min-entropy provides the optimal bound.

\section{Thermal noise scenario: estimates for the secret-key rate for N-PSK protocol}

In this section we extend the finite-size analysis of Sec.~\ref{keyrateBPSK} to the case where the physical channel between Alice and Bob is a thermal-noise channel, described in Appendix \ref{sec:thermal}. The thermal-noise channel with transmittance $\eta$ and mean thermal photon number $\mu$ can be represented by mixing Alice’s signal mode with one arm of a two-mode squeezed vacuum state $|\Psi\rangle_{E_1 E_2}$ at a beam splitter of transmittance $\eta$. Bob receives the transmitted mode and performs homodyne or heterodyne detection, while Eve keeps both environmental modes. In this scenario Eve’s relevant quantum memory is therefore the two-mode system $E = E_1 E_2$. 

An important difference with respect to the pure-loss case is that, in the presence of thermal noise, the conditional states $\rho_{E|x}$ are no longer pure. This prevents us from applying the orthonormalization procedure exploited in Sec.~\ref{sec:CondEnt} to obtain closed expressions for the conditional entropies. Instead, all the entropic functionals inside our bounds have been evaluated  numerically. Since    Eve's Hilbert space is now infinite-dimensional, we have to introduce a photon-number cutoff and work with a truncated Fock basis for each of the two modes $E_1$ and $E_2$. In practice, we fix a cutoff $n_c$ and represent all the states on the finite-dimensional subspace
\begin{align}
    \mathcal{H}_E^{(n_c)} = \mathrm{span}\big\{\,|n_1\rangle_{E_1} \otimes |n_2\rangle_{E_2} : 0 \leq n_1,n_2 \leq n_c \,\big\},
\end{align}
so that $\dim \mathcal{H}_E^{(n_c)} = (n_c+1)^2$. By contrast, the classical conditional entropy $H_N(Y|X)$ that describes error correction remains available in closed form, since the classical channel $p(y|x)$ in the thermal case is still given explicitly in Appendix \ref{sec:thermal}. In particular, for BPSK we  generalized homodyne statistics in presence of thermal noise in Eq.~(\ref{thermalhomo}), while for QPSK we use the expressions in Eq.~(\ref{thermalhetero}), which generalize Eq.~(\ref{matrix}) by replacing $P_{\pm}$ with $P_\pm^{(\mu)}$.

In our bounds for the key rate we have to account for the fact that all these entropic quantities are computed on a truncated version of Eve’s Hilbert space. To obtain valid bounds for the original infinite-dimensional model, we combine our finite-size analysis with a dimension-reduction argument in the spirit of the approach developed in Ref.~\cite{FlorianPRXQ2023}. Given a cutoff $n_c$, this argument provides an upper bound $\omega$ on the weight of the state outside the truncated subspace $\mathcal{H}_E^{(n_c)}$. In the thermal-noise scenario we focus on two finite-size estimators of the secret-key rate, one based on the sandwiched Rényi conditional entropy and one based on the asymptotic equipartition property (AEP). They are structurally identical to the bounds $r_{\epsilon''}^{S}$ and $r_{\epsilon''}^{AEP}$ introduced in Eqs.~(\ref{ratesandBPSK})–(\ref{rateAEPBPSK}) for the pure-loss case, but they include an additional penalty term arising from the dimension reduction. Moreover, for the sandwiched bound we now use the non-optimized entropy $\tilde H_a^\downarrow(Y|E)$ rather than its optimized counterpart $\tilde H_a^\uparrow(Y|E)$, in order to avoid an extra optimization over states on the enlarged two-mode system $E_1E_2$, which would be numerically demanding. With these conventions, for a generic $N$-PSK protocol we consider two different bounds.
\begin{enumerate}
    \item From the Sandwiched Rényi $a$-entropy $\tilde{H}_a^\downarrow(Y|E)$, ($a>1$) (it represents a relaxation (\ref{RateSand}) to non-optimized scenario),
    \begin{align}
    \label{eq:rSdown_omega}
    r_{\epsilon'',\,\omega}^{SD} \ge\;&
        \tilde{H}_a^\downarrow(Y|E)
        +\frac{ 1+ 2 \log{\epsilon'}}{n}  
        -\frac{g(\epsilon)}{n(a-1)}\\
        &-H_N(Y|X)\notag
        -\Delta(\omega),
\end{align}
\item From the AEP (\ref{RateAEP}),
\begin{align}
\label{eq:rAEP_omega}
    r_{\epsilon'',\,\omega}^{AEP} \ge\;&
        H(Y|E) 
        +\frac{ 1+ 2 \log{\epsilon'}}{n}
        - \frac{\delta(\epsilon)}{\sqrt{n}} \\
        &-H_N(Y|X)\notag
        -\Delta(\omega),
\end{align}
\end{enumerate}
The term
\begin{equation}
    \Delta(\omega) = \sqrt{\omega}\,\log_2(N)
    +\big(1+\sqrt{\omega}\big)\,
      h_2\!\left(\frac{\sqrt{\omega}}{1+\sqrt{\omega}}\right)
    \label{eq:Delta_omega}
\end{equation}
where $h_2$ is the Shannon binary entropy, is the correction induced by the dimension reduction technique (see Appendix~\ref{sec:weight}).

\subsection{Estimates for the secret-key rate for BPSK protocol}

In Fig.~\ref{fig:comparison_distance_BPSK_d} 
we show a comparison between $r_{\epsilon'',\,\omega}^{SD}$ and $r_{\epsilon'',\,\omega}^{AEP}$ computed for some fixed values of the block size $n=10^7,10^9$, against the transmission distance $d$ measured in $km$. We recall that once the specific attenuation $\zeta$ of the transmission medium is specified in $dB/km$, then the transmittance $\eta$ is given by 
\begin{align}
    \eta = 10^{-\zeta d/10}.
    \label{DistTrans}
\end{align} 
Here we set $\zeta = 0.2 ~dB/km$, and again $\epsilon = \epsilon' = 10^{-8}$. We set $\alpha = 0.5$ for the coherent amplitude and $a=1.05$ as the Rényi parameter for $r_{\epsilon'',\,\omega}^{SD}$. All the key rates are computed for an excess noise of $\xi = 0.01$. It is manifest how the key rate computed through Rényi entropy is much more robust against losses. If we consider a block size $n=10^7$ $r_{\epsilon'',\,\omega}^{AEP}$ is already vanishing at $d\sim60~km$, while $r_{\epsilon'',\,\omega}^{SD}$ endures until $d\sim 140~km$. For a greater block size $n=10^9$, $r_{\epsilon'',\,\omega}^{AEP}$ is able to keep up with $r_{\epsilon'',\,\omega}^{SD}$ until around $d\sim 80~km$, but then experiences a drop soon after $d\sim 100~km$. We remark that for the same block size, the $r_{\epsilon'',\,\omega}^{SD}$ is still non-vanishing even for $d>200~km$. An equivalent study is shown in Fig.~\ref{fig:comparison_eta_BPSK} where the key rates are plotted against the transmittance $\eta$ in place of the distance $d$. Here we take into consideration an even smaller block size of $n=10^5$; we can note that, as we could expect from the pure loss case, when the block size decreases, the advantage of the Rényi estimator w.r.t. the AEP one becomes more and more manifest.

Complementary results are shown in Fig.~\ref{fig:comparison_distance_BPSK_blocksize}, where the two different key rates in Eqs.~(\ref{eq:rSdown_omega}) and (\ref{eq:rAEP_omega}) are plotted against the block size $n$ for a fixed value of distance $d=10~km$; all the other parameters of the protocol are set like in the former analysis. Here a direct comparison with the pure-loss scenario can be done by looking at Fig.~\ref{rateBPSK}. Looking at the region of large block sizes first, we observe that the key rates saturates around $r_{\epsilon'', \omega}\sim 0.14$, which is less than a third of the value obtained for a pure-loss channel; the fact that asymptotically the values of $r_{\epsilon'',\,\omega}^{SD}$ and $r_{\epsilon'',\,\omega}^{AEP}$ do not converge can be ascribed to the lack of the optimization on the Rényi parameter $a$, that should decrease more and more as $n\to\infty$. Looking at the region of small block sizes, we note that the thermal noise forces the key rates to drop earlier to zero; for $r_{\epsilon'',\,\omega}^{SD}$ it happens around $n\sim 10^4$ which clearly represents a worsening w.r.t. the loss-only scenario.

\subsection{Estimates for the secret-key rate for QPSK protocol}

In this section we present the results regarding QPSK protocol, comparing them with Ref~\cite{FlorianPRXQ2023}. 
Though we do not perform the comprehensive finite-size analysis presented by Kanitschar et al., we adopt their testing ratio parameter $r_{test}$ to enable a consistent comparison with their results; specifically, $r_{test}$ represents the fraction of the total transmitted signals $n$ that are sacrificed for statistical testing, such as the energy and acceptance tests, rather than being utilized for key generation. In Eqs.~(\ref{eq:rSdown_omega}) and (\ref{eq:rAEP_omega}) this can be visualized as an overall multiplicative factor $(1-r_{test})$ for all the quantities on the RHS, but $\frac{1}{n} \left( 1+2\log\epsilon' \right)$. 
In Fig.~\ref{fig:comparison_distance_QPSK_blocksize} we show the secret-key rates vs the distance, taking into consideration only $r_{\epsilon'',\,\omega}^{SD}$, for fixed block sizes; like in BPSK, we choose fixed values for the protocol parameters $\alpha=0.85$ and $a=1.05$; the test ratio parameter is set at $r_{test} = 0.1$. One reference curve from~\cite{FlorianPRXQ2023} corresponding to $n=10^9$ is also drawn for comparison.
The advantage of the Rényi estimator is great; by the way our work relies on the analysis of a specific known gaussian attack; in contrast, Kanitschar et al.~minimize the functional of interest subject to experimental constraints—namely, the observed moments of the observables and the bounds imposed by the energy test and it is this distinct optimization methodology that primarily accounts for the difference in key rates. This same argument is valid in the analysis of Fig.~\ref{fig:comparison_distance_QPSK_blocksize}, where both $r_{\epsilon'',\,\omega}^{SD}$ and $r_{\epsilon'',\,\omega}^{AEP}$ are shown against the block size, for a fixed distance of $d=10~km$. Here the reference curve is not plotted, but with the same values of the protocol parameters, an asymptotic value of $r_{\epsilon'',\omega}\sim 0.13$ is found in~\cite{FlorianPRXQ2023} (fig.~2); in contrast our results show slightly higher values in the range $[0.17,0.185]$. As for the small block size regime, we note that our secret-key rates endure better, while the ones in~\cite{FlorianPRXQ2023} drop to zero around $n\sim10^9$. Finally in Fig.~\ref{fig:comparison_eta_QPSK}, we show the secret-key rates against the transmittance; the trends are very similar to the BPSK case.

\begin{figure}
    \centering
    \subfloat[BPSK ($\alpha=0.5,\,a=1.05$).]{
        \includegraphics[width=\linewidth]{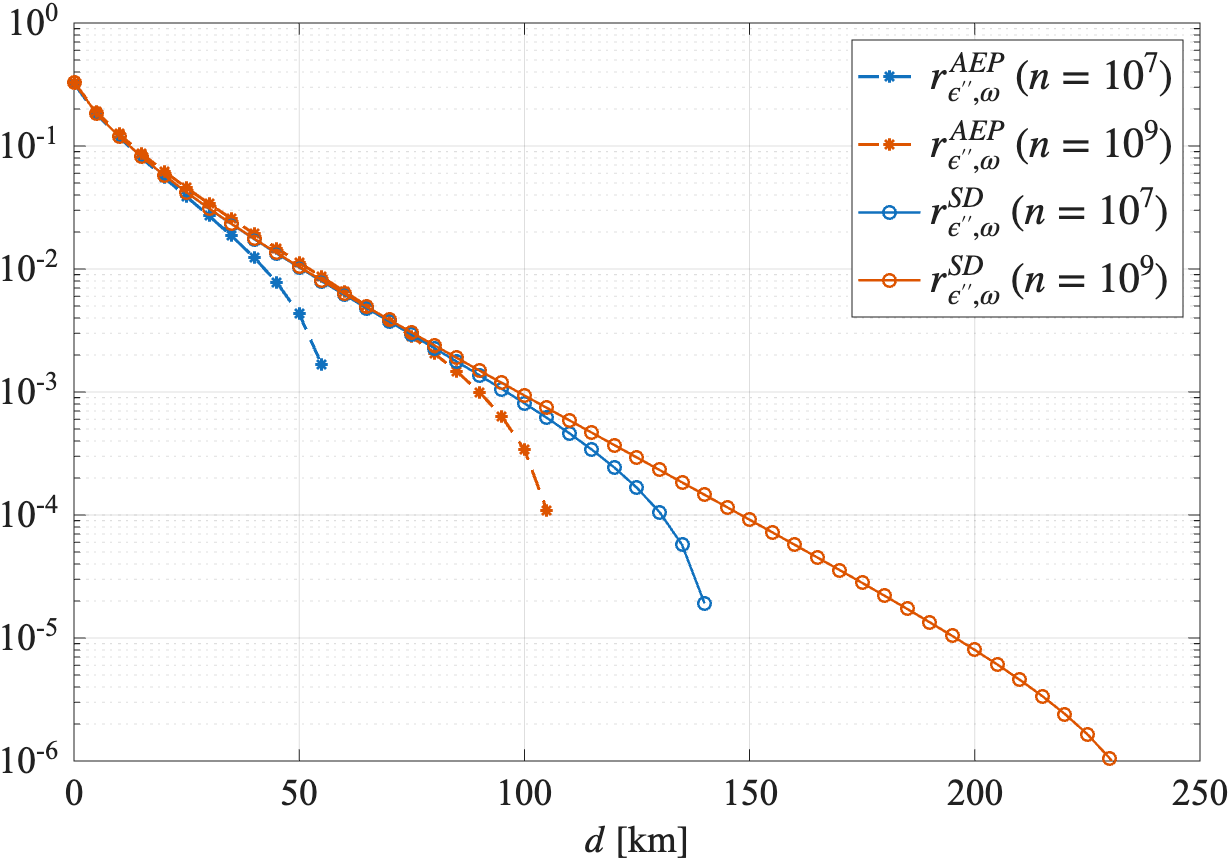}
        \label{fig:comparison_distance_BPSK_d}
    }\\
    \subfloat[QPSK protocol ($\alpha=0.85,\,a=1.05$, $r_{test}=0.1$).]{
        \includegraphics[width=\linewidth]{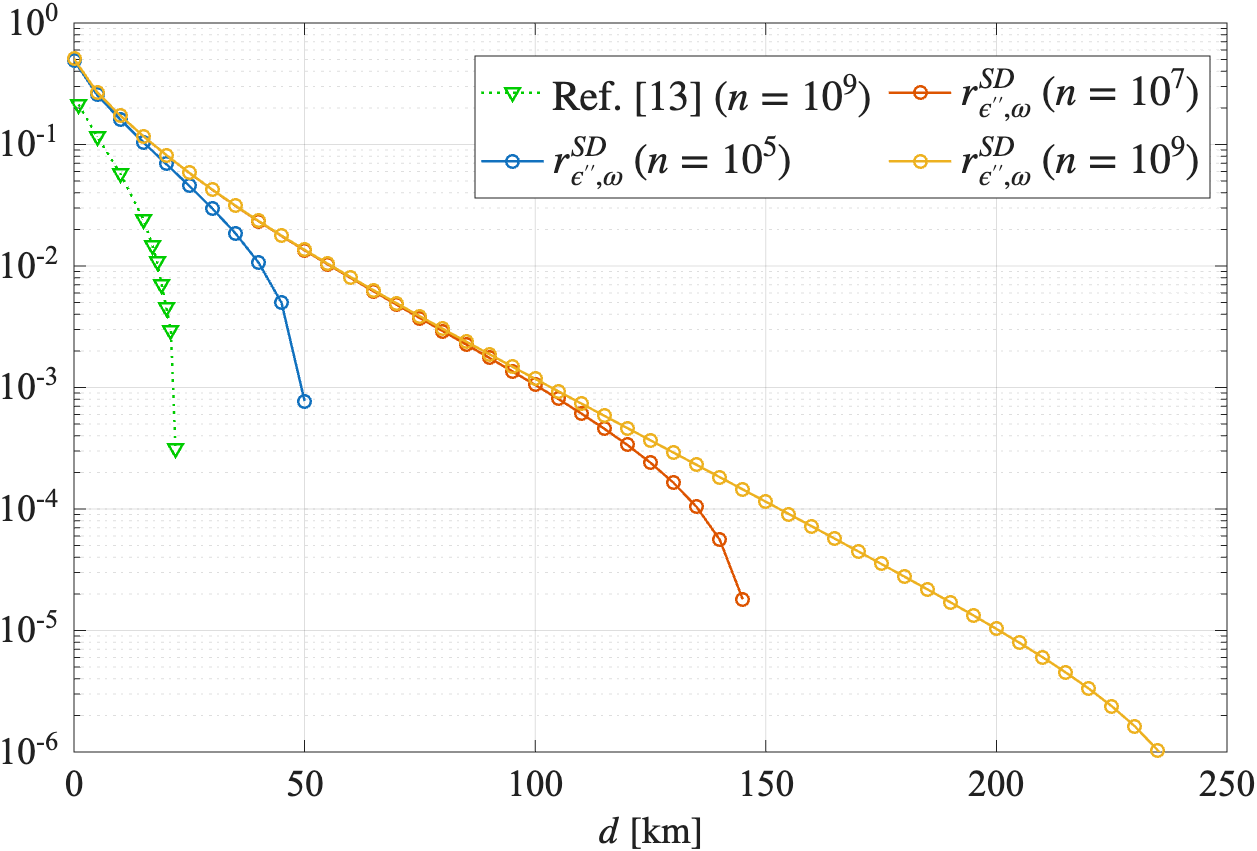}
        \label{fig:comparison_distance_QPSK_d}
    }
    \caption{\justifying Secret-key rates computed on the CQ states at $\alpha$ and $a$ on $\rho_{YE}$ in Eqs.~(\ref{rhoYEnoiseB}) and (\ref{rhoYEnoiseQ}) respectively for BPSK in \ref{fig:comparison_distance_BPSK_d} and in \ref{fig:comparison_distance_QPSK_d} for QPSK protocol, in term of the distance. For QPSK an additional curve is added that reproduce the results in \cite{FlorianPRXQ2023}.}
    \label{fig:comparison_distance}
\end{figure}

\begin{figure}
    \centering
    \subfloat[BPSK ($\alpha=0.5,\,a=1.05$).]{
        \includegraphics[width=\linewidth]{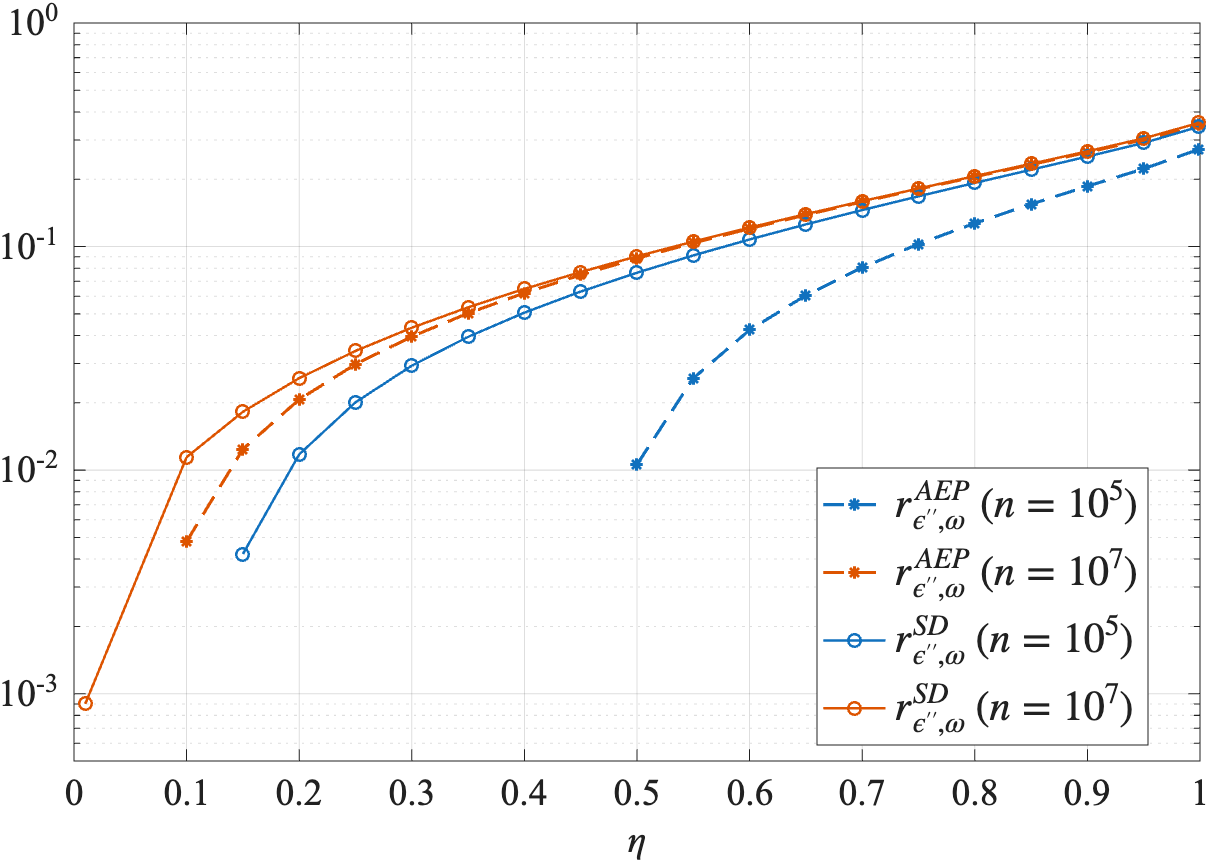}
        \label{fig:comparison_eta_BPSK}
    }\\
    \subfloat[QPSK protocol ($\alpha=0.85,\,a=1.05$).]{
        \includegraphics[width=\linewidth]{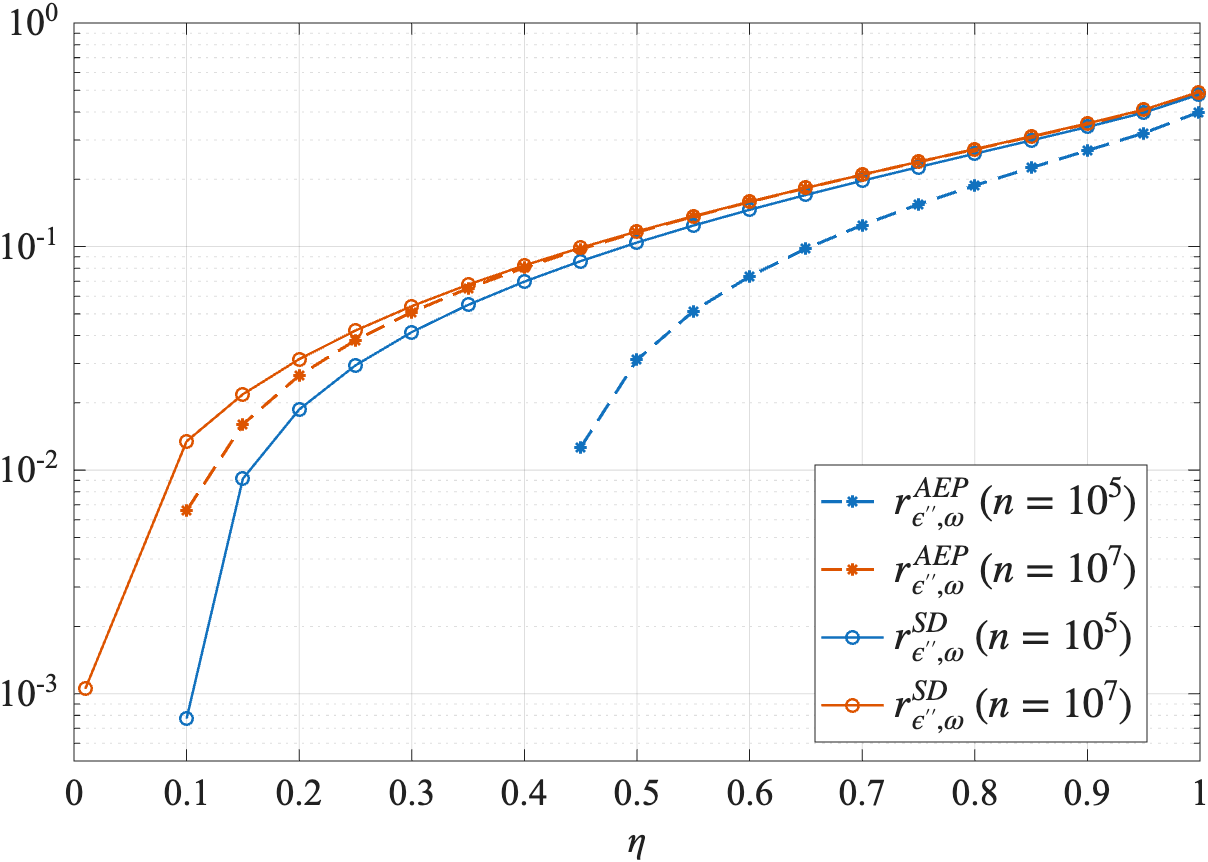}
        \label{fig:comparison_eta_QPSK}
    }
    \caption{\justifying  Secret-key rates computed on the CQ states at $\alpha$ and $a$ on $\rho_{YE}$ in Eqs.~(\ref{rhoYEnoiseB}) and (\ref{rhoYEnoiseQ}) respectively for BPSK in \ref{fig:comparison_eta_BPSK} and in \ref{fig:comparison_eta_QPSK} for QPSK protocol, in term of the transmittance $\eta$.}
    \label{fig:comparison_eta}
\end{figure}

\section{Related Work}\label{relatedW}

Recently, DM CV QKD has made significant progress both experimentally and theoretically. High-speed tests using larger constellations have helped clarify how performance changes with distance, measured in terms of secret bits per channel use. 
For example, Ref.~\cite{Pan2022OL} shows asymptotic secret-key rates at $1$~Gbaud system with discrete-Gaussian 64/256-QAM of about ($\sim300$~Mbit/s, $5$km) to $\sim7$~Mbit/s,$50$km);
%
Using a smaller constellation (16 symbols, 2~rings) at $2.5$~Gbaud, Ref.~\cite{Tian2023OL} obtain ($49$~Mbit/s, $25$ km) to ($2.11$~Mbit/s, $80$~km).
%
Roumestan \emph{et al.}~demonstrate DM formats and report \emph{tens of Mbit/s} at $25$~km~\cite{Roumestan2022arXiv}. QPSK realization has recently achieved \emph{composable, finite-size} security against collective attacks over $20$~km with $N\simeq2.3\times10^{9}$ signals and a positive key fraction of $11.04\times10^{-3}$~bit/symbol using heterodyne detection and an advanced finite-size proof stack~\cite{TobiasDM2024}. These results collectively set realistic targets for protocol design at metropolitan scales.

\emph{Theory side.} On the finite-size composable front with realistic detection, Ref.~\cite{Lupo2022} establishes composable security for DM CV QKD under nonideal heterodyne detection (collective attacks), delivering tight numerical bounds without Hilbert-space truncations.
For phase-shift keying specifically,  Ref.~\cite{kanitschar2022optimizing} presents an optimized scheme (including post-selection), providing tight asymptotic benchmarks relevant to both homodyne and heterodyne decoding, while Ref.~\cite{Liu2021PRXQ} analyzes QPSK showing enhanced tolerance to noise. 
Beyond collective attacks, finite-size security of QPSK DM CV-QKD is established in Ref.~\cite{AcinQuantum2024} against \emph{coherent} attacks, and in Ref.~\cite{PrimaatmajaNPJQI2025,PascualGarciaPRA2025} against sequential attacks.
%
%
For asymptotic setting, a tight security analysis for DM is presented in Ref.~\cite{Ghorai2019}, and an SDP-based asymptotic framework for DM (including QPSK) is discussed in~\cite{Lutk2019}. An explicit asymptotic key-rate formula for \emph{arbitrary} DM is in Ref.~\cite{Denys2021}. 

%
In Ref.~\cite{FlorianPRXQ2023}, an energy test combined with dimension reduction was specifically adapted for DM CV QKD, enabling finite-size AEP key rates for QPSK that remain positive under realistic experimental conditions (reported up to about $72$~km for i.i.d.~collective attacks (see their Figs.~2 and~3).
To make a fair comparison, in Appendix~\ref{sec:thermal} we explain how to turn the passive into a thermal attack, a.k.a.~\textit{excess noise} in a way similar to Ref.~\cite{FlorianPRXQ2023}, with one key difference: in our case, the expressions for $\rho_{YE}$ in Eq.~\eqref{rhoYEnoiseB} and~\eqref{rhoYEnoiseQ} are assumed to be known.
A direct comparison is still not possible because we do not use the optimal state that maximizes Eve’s information, nor do we optimize over the parameters $a$ and $\alpha$.
Nevertheless, the results shown in Figs.~\ref{fig:comparison_distance} and~\ref{fig:comparison_blocksize} for both BPSK and QPSK reveal an interesting advantage when the sandwich conditional entropy $\tilde H^{\downarrow}_a$ is applied. 
Since we are not considering the worst-case scenario, our key rates are naturally higher than those in  Ref.~\cite{FlorianPRXQ2023}. However, we expect that the observed difference between AEP and sandwich rates will persist even under worst-case conditions. This is an initial observation that warrants further investigation.

\begin{figure}
    \centering
    \subfloat[BPSK ($\alpha=0.5,\,a=1.05, \,d=10~km$).]{
        \includegraphics[width=\linewidth]{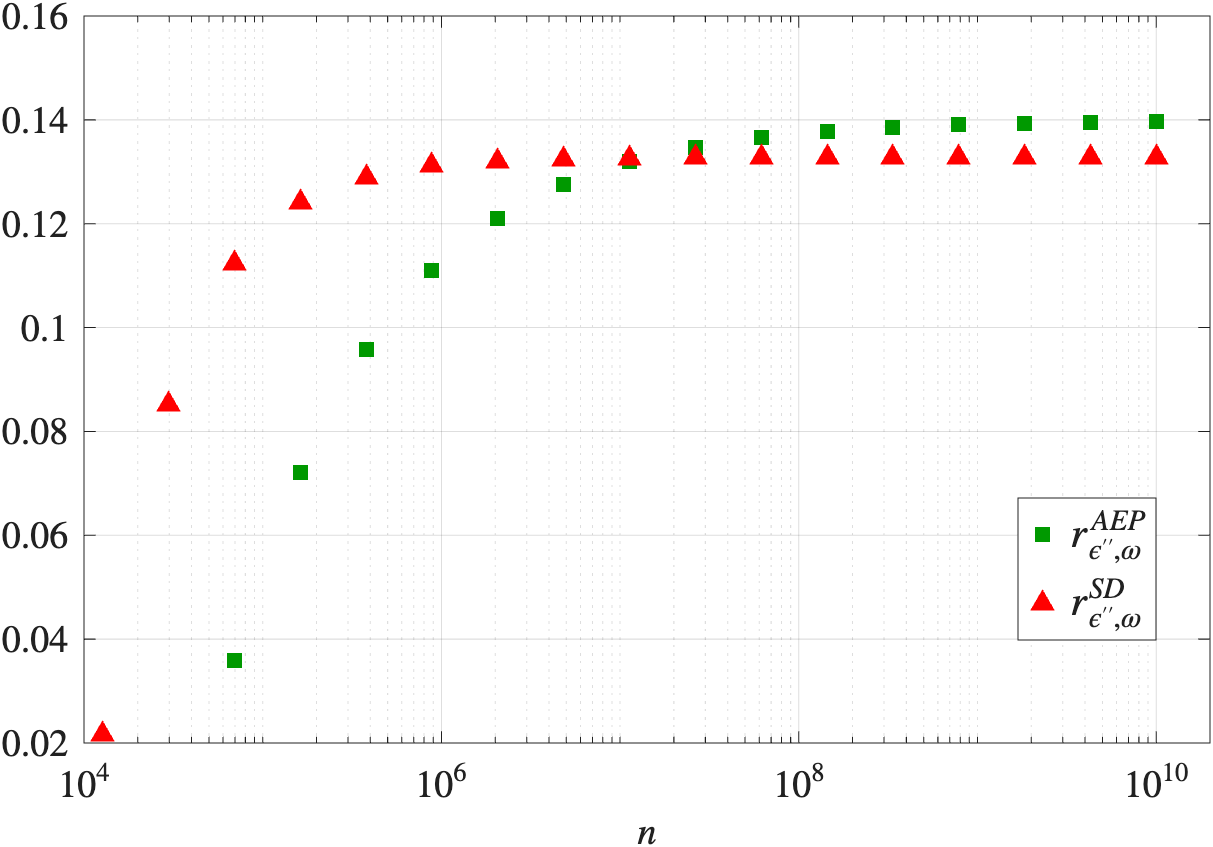}
        \label{fig:comparison_distance_BPSK_blocksize}
    }\\
    \subfloat[QPSK protocol ($\alpha=0.85,\,a=1.05, \,d=10~km$, $r_{test}=0.1$).]{
        \includegraphics[width=\linewidth]{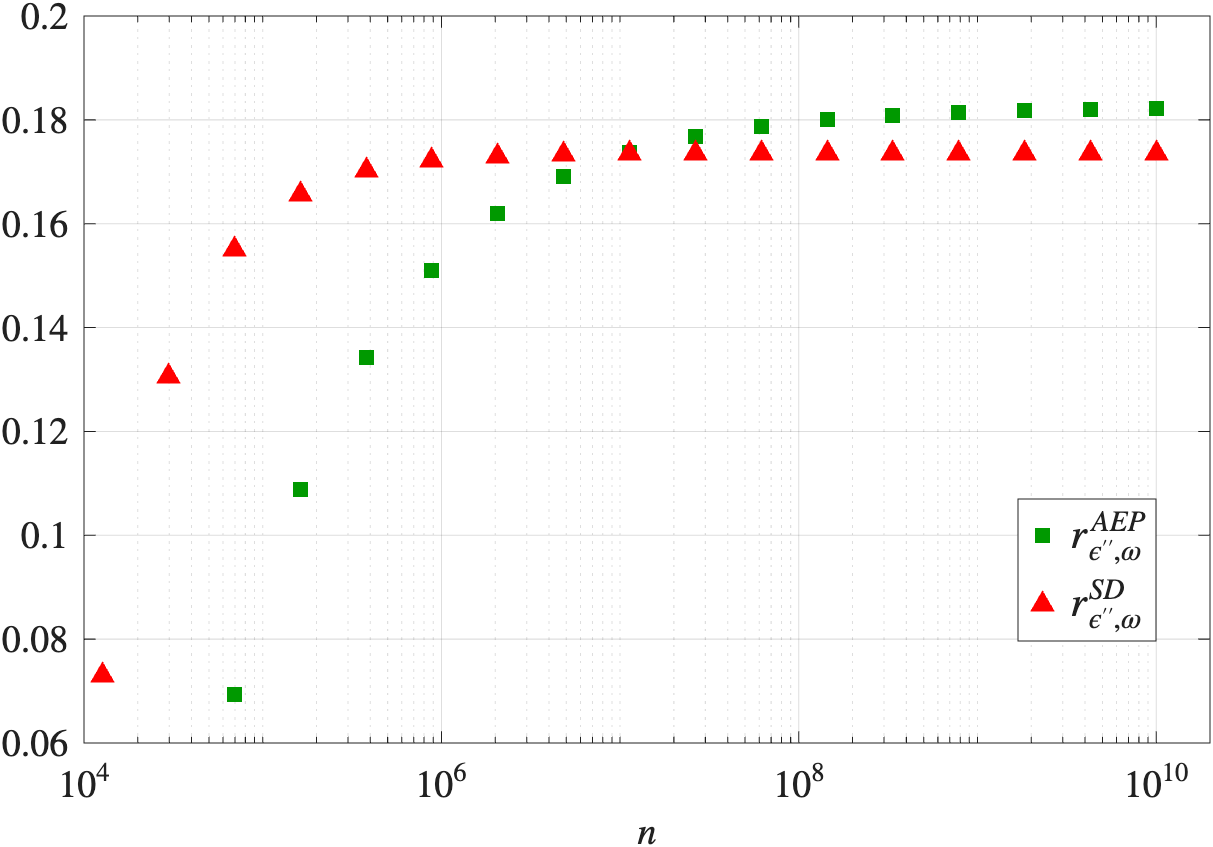}
        \label{fig:comparison_distance_QPSK_blocksize}
    }
    \caption{\justifying Secret-key rates computed on the CQ states at $\alpha$ and $a$ on $\rho_{YE}$ in Eqs.~(\ref{rhoYEnoiseB}) and (\ref{rhoYEnoiseQ}) respectively for BPSK in \ref{fig:comparison_distance_BPSK_d} and in \ref{fig:comparison_distance_QPSK_d} for QPSK protocol, in term of the block size $n$.}
    
    \label{fig:comparison_blocksize}
\end{figure}
\section{Conclusion}
In this work, we have computed and compared several quantum conditional entropies that play a fundamental role in establishing the secrecy of QKD.
Our work focuses on CV QKD protocols exploiting discrete modulation (BPSK and QPSK) of coherent states, and decoding by homodyne or heterodyne detection, a.k.a.~double homodyne.
We assume passive, collective attacks in the pure-loss scenario in Fig.~\ref{fig:overallrate}, and the thermal noise channel 
in Fig.~\ref{fig:comparison_distance},~\ref{fig:comparison_eta}
and~\ref{fig:comparison_blocksize}.
%
Our results do provide meaningful ball-park figures, as Gaussian attacks --- which model linear loss and Gaussian noise in optical fibers --- are commonly used to benchmark quantum communication protocols.

We have considered four notions of quantum conditional entropies, the Petz-Rényi entropy $H_a^\downarrow$, the sandwiched Rényi entropy $\tilde H_a^\downarrow$ and their optimized versions $ H_a^\uparrow$ and $\tilde H_a^\uparrow$~\cite{Petz1986,Wilde2014,Tomamichel2014}, and three different ways to estimate the finite-size secret-key rate. In particular, $r_{\epsilon''}^S$ exploits directly the bound on the smooth min-entropy based on the optimized sandwiched Rényi entropy~\cite{tomamichel2012framework,EAT2020}; $r_{\epsilon''}^B$ follows immediately from the first one, when one bounds by a continuity argument Renyi entropies with the Von Neumann one like in \cite{improved2ndorder}; the third estimate of the secret-key rate $r_{\epsilon''}^{AEP}$ is obtained from a more traditional approach based on asymptotic equipartition property (AEP) \cite{tomamichel2012framework}.


We computed explicit bounds on the secret-key rates for both BPSK and QPSK protocols. The approach based on the optimized sandwiched Rényi entropy gives the best key-rate estimate, especially for very small block sizes.
This quantity has also been very recently considered by other authors in the context of discrete-variable QKD~\cite{Cai2025, Tan_arXiv2025}. To the best of our knowledge, our work represents the first application within CV QKD.
As a further development, and to make our analysis more complete and reliable, one needs to depart from the assumption of lossy or thermal channels with a known attenuation factor or excess noise. In order to do so, a dual representation of the quantum conditional entropy is required. 
For completeness, and in view of this future development, in Appendix~\ref{dual} we provide the dual expressions of the quantum conditional entropies that are the basic building block to obtain reliable, numerical estimates of the secret-key rates. These entropic functionals can also be exploited within device-independent QKD, as shown in recent works on the connection between Rényi entropies, Bell inequality violations and security proofs~\cite{Hahn2024,Hahn2025}. Moreover, our scheme is naturally compatible with current experimental platforms, and could be tested in setups of the kind reported in Sec.~\ref{relatedW}, thereby linking the theoretical framework with practical implementations.
The data that support the findings of this article are openly available \cite{gioscaNPSK}.

\begin{acknowledgments}
This work has received support by the European Union's Horizon Europe research and innovation programme under the Project ``Quantum Secure Networks Partnership'' (QSNP, Grant Agreement No.~101114043);
and by INFN through the project ``QUANTUM''. The authors acknowledge support from Innovation Fund Denmark (AccessQKD, grant agreement no. 4356-00004B).
\end{acknowledgments}

\appendix

\section{Entropic bounds for the finite-size key rates}
\label{sec2}

Here we briefly review the definitions of the Petz-Rényi and sandwiched Rényi quantum conditional entropies. 

For $a \in (0,1)\cup(1,+\infty)$, let $\rho$ and $\sigma$ be two quantum states with $\rho \ll \sigma$ (here $\rho \ll \sigma$ means that the support of $\rho$ is contained in the support of $\sigma$). 
We consider two different notions of Rényi divergence between states $\rho$ and $\sigma$:  
\begin{itemize}
    \item The Petz-Rényi divergence \cite{Petz1986}:
    \begin{align}
    D_a(\rho \| \sigma) = \frac{1}{a - 1}\log \mathrm{tr}\!\left( \rho^a \sigma^{1-a} \right) \, ;
    \end{align}
    \item The sandwiched Rényi divergence \cite{Wilde2014,Tomamichel2014}:
     \begin{align}
    \tilde{D}_a(\rho \| \sigma) = \frac{1}{a - 1}\log \mathrm{tr}\!\left[ \left( \sigma^{\frac{1-a}{2a}} \rho \, \sigma^{\frac{1-a}{2a}} \right)^a \right] \, .
    \label{SandDownDIV}
     \end{align}
\end{itemize}
We denote always with $\log$ the logarithm base $2$, while $\ln$ is the natural logarithm.  From these divergences, one defines the corresponding conditional Rényi entropies~\cite{Muller}. Let $A$ and $B$ be two quantum systems, and denote by $\mathcal{S}(A)$ the set of density operators describing the states of system $A$. 
For a bipartite state $\rho_{AB}$ with reduced state $\rho_B = \mathrm{tr}_A(\rho_{AB})$, we set:  
\begin{align}
H_a^\downarrow(A|B) 
& = 
    -D_a
    \left(
    \rho_{AB} \big\| 
    \mathbb I_A \otimes \rho_B
    \right) 
     \nonumber\\
& = \frac{1}{1-a} 
    \log \mathrm{tr}
    \left( \rho_{AB}^a \, \rho_B^{1-a} 
    \right),
    \label{PetzDown} 
   \\
\tilde{H}_a^\downarrow(A|B) &=
    -\tilde{D}_a
    \left(\rho_{AB} \big\| 
    \mathbb I_A \otimes \rho_B
    \right) 
     \nonumber \\
& = \frac{1}{1-a} 
    \log \mathrm{tr}
    \left[ 
    \left( \rho_B^{\frac{1-a}{2a}} \rho_{AB} \, \rho_B^{\frac{1-a}{2a}} 
    \right)^a 
    \right].\label{SandDown}
\end{align}

In the above expressions, $\rho_B$ is implicitly understood as $\mathbb I_A \otimes \rho_B$ when it appears inside the trace.  
A natural generalization arises by optimizing over all possible states $\sigma_B \in \mathcal{S}(B)$.  
This leads to the “$\uparrow$”-versions of the entropies:  
\begin{align}
H_a^\uparrow(A|B) 
& = \sup_{\sigma_B \in \mathcal{S}(B)} -D_a\!\left(\rho_{AB} \big\| \mathbb I_A \otimes \sigma_B\right) 
 \label{PetzUp} \\
\tilde{H}_a^\uparrow(A|B) &= \sup_{\sigma_B \in \mathcal{S}(B)} -\tilde{D}_a\!\left(\rho_{AB} \big\| \mathbb I_A \otimes \sigma_B\right) 
 \label{SandUp} 
\end{align}
It is known that $D_a(\rho \| \sigma) \geq \tilde{D}_a(\rho \| \sigma)$~\cite{Datta}.  
Combined with the trivial relations between $\uparrow$ and $\downarrow$ versions of the same type, this yields the following monotonicity scheme (the arrows go from the larger quantity to the smaller one) valid for all states $\rho_{AB}$ \cite{BertaIEEE}:  
\begin{equation}
\begin{tikzcd}
H_a^{\uparrow}(A|B) \arrow[d] & \tilde{H}_a^{\uparrow}(A|B) \arrow[l] \arrow[dl] \arrow[d] \\
H_a^{\downarrow}(A|B) & \tilde{H}_a^{\downarrow}(A|B) \arrow[l]
\label{monotone}
\end{tikzcd}
\end{equation}

Moreover, both $D_a(\rho \| \sigma)$ and $\tilde{D}_a(\rho \| \sigma)$ are monotonically increasing in $a$, then all the different conditional Rényi entropies will be monotonically decreasing in $a$~\cite{Tomamichel2014}. Taking the limit $a\to1$ in a suitable way~\cite{Muller} one can recover the relative entropy
$
    D_1(\rho\|\sigma) 
    = \tilde{D}_1(\rho\|\sigma) 
    = D(\rho\|\sigma) 
    = \mathrm{tr}
    \left[
    \rho(\log\rho-\log\sigma)
    \right]
$
and thus define the well-known von Neumann conditional entropy
\begin{align}
    H(A|B) &= 
    -D(\rho_{AB}\|\mathbb I_A\otimes \rho_B)
    \label{VonNeumann}
    \\
    &=\sup_{\sigma_B\in\mathcal S (B)}
    -D(\rho_{AB}\|\mathbb I_A\otimes \sigma_B)
\end{align}
Hence for $a \to 1$ all the different Rényi entropies defined earlier converge to the same quantity. Furthermore, in the limit $a\to\infty$~\cite{Muller} two different definitions of conditional min-entropy~\cite{RennerThesis,tomamichel2012framework} emerge
\begin{align}
    \tilde{H}_{\mathrm{min}}^\downarrow(A|B) 
    &=-\tilde D_\infty (\rho_{AB}\|\mathbb I_A\otimes \rho_B) \\
    & = \sup_\lambda
    \left\{
    \lambda\in \mathbb R |
    \rho_{AB}\le2^{-\lambda}
    \mathbb I_A\otimes \rho_B 
    \right\}\nonumber\\
    &= -\log
    \left\|
    (\mathbb I_A\otimes \rho_B )^{-\frac12}\;
    \rho_{AB}\;
     (\mathbb I_A\otimes \rho_B )^{-\frac12}
    \right\|_\infty \,,
\end{align}
and its optimization reads
\begin{align}
    \tilde{H}_{\mathrm{min}}^\uparrow(A|B) 
    &= \sup_{\sigma_B\in \mathcal S(B)}-\log
    \left\|
     \sigma_B ^{-\frac12}
    \rho_{AB}\;
     \sigma_B ^{-\frac12}
    \right\|_\infty
    \, .\label{HminUP}
\end{align}
Here again $\sigma_B$ is implicitly understood as $\mathbb I_A \otimes \sigma_B$ inside the norm. 
Remarkably, the data processing inequality (DPI) has not been proven to hold for the Petz-Rényi divergence for $a>2$ \cite{Petz1986,Muller,Audenaert}, then it has not been taken into consideration for a possible definition on the conditional min-entropy. Now we recall the definition of \textit{smooth} conditional min-entropy \cite{tomamichel2012framework}. Let $\rho_{AB}\in \mathcal S(AB)$ and $\epsilon>0$. We define the $\epsilon$-ball of density operator around $\rho_{AB}$ as
\begin{align}
    \mathcal{B}^\epsilon(\rho_{AB}) & = 
    \left\{
    \tau_{AB}\in \mathcal S(AB) |\,\sqrt{1-F(\rho_{AB},\tau_{AB})^2}\le\epsilon
    \right\},
\end{align}
with $F(\rho_{AB},\tau_{AB})$ the quantum fidelity. The smoothed version of Eq. \eqref{HminUP} is
\begin{equation}
     \tilde H^{\uparrow \epsilon}_{\mathrm{min}}(A|B) 
    = \sup_{\substack{\tilde \rho_{AB}\in \mathcal B^\epsilon(\rho_{AB}) \\ \sigma_B\in\mathcal S(B)}}
   -\log
    \left\|
    \sigma_B^{-\frac12}\;
    \tilde \rho_{AB}\;
     \sigma_B^{-\frac12}
    \right\|_\infty
    \, .
\end{equation}

\section{Details on computation of conditional entropies for PSK protocol}\label{renyiproof}

In this Appendix we obtain explicit expressions for the quantum conditional entropies of a CQ state of the form (\ref{CQdef}). In particular, we focus on the CQ states arising from BPSK and QPSK protocols in reverse reconciliation and under the assumption of passive attacks.

The Petz-Rényi conditional entropy \eqref{PetzDown} of the CQ state reads
\begin{align}
    H_a^\downarrow& (Y|E) 
     =
    \frac{1}{1-a}
    \log 
    \mathrm{tr}
    \left\{
    \left( \sum_{y=0}^{N-1} p_y |y\rangle\langle y| \otimes \rho_{E|y} \right)^a \rho_E^{1-a} 
    \right \}\\
    & =
    \frac{1}{1-a}
    \log 
    \mathrm{tr}
    \left\{ \left(
    \frac{1}{N^a}\sum_{y=0}^{N-1} |y \rangle\langle y| \otimes \rho_{E|y}^a \right)
    \rho_E^{1-a} 
    \right \} \label{ptezl2} \\
 &=
    \frac{1}{1-a}
    \log \left\{
    \frac{1}{N^a}\sum_{y=0}^{N-1} \,
    \mathrm{tr}
    \left( 
     \rho_{E|y}^a \;\rho_E^{1-a} 
    \right )
    \right\} \, .
    \label{petzl4APP}
\end{align}
In line \eqref{ptezl2} we put $p_y=1/N$ and exploit the fact that the matrix which is raised to the power $a$ is block-diagonal, then the exponentiation commutes with the tensor product. Indeed if we consider $M = \bigoplus_{y=0}^{N-1}M_y$, then $M^a = \bigotimes_{y=0}^{N-1}M^a_y$, and in our case we have $\sum_{y=0}^{N-1} p_y \, \ket{y}\bra{y} \otimes \rho_{E|y}
= \bigoplus_{y=0}^{N-1} \bigl(p_y \rho_{E|y}\bigr)$. Moreover in line \eqref{petzl4} we use the property for the trace of a Kronecker product $\mathrm{tr}(\rho\otimes \sigma) = \mathrm{tr}(\rho)\mathrm{tr}( \sigma)$ and the linearity of the trace.

We also notice that for both the examples of BPSK and QPSK discussed here (see Sections \ref{subsec:BPSK} and \ref{subsec:QPSK} below), and under the assumption of passive attacks, the conditional states $\rho_{E|y}$ are obtained through the action of a symmetry group. 
Therefore, we have
\begin{align}
    \rho_{E|y \oplus t} = U_t \, \rho_{E|y} \, U_t^\dag \, ,
    \label{symmetry}
\end{align}
where $\{ U_t \}_{t=0, \dots, N-1}$ is a finite group of unitary matrices, for $t=0, \dots, N-1$, and $\oplus$ denotes summation modulo $N$.
This also implies that the average state $\rho_E = \frac{1}{N} \sum_{y=0}^{N-1} \rho_{E|y}$ is invariant under the action of $\{ U_t \}$.
Therefore, we can write for the Petz-Rényi conditional entropy:
\begin{align}
     H_a^\downarrow (Y|E) 
    & =
    \frac{1}{1-a}
    \log \left\{
    \frac{1}{N^a}
    \sum_{y=0}^{N-1}   \,
    \mathrm{tr}
    \left( 
     \rho_{E|y}^a \;\rho_E^{1-a} 
    \right )
    \right\}\\
& =
    \frac{1}{1-a}
    \log \left\{
    \frac{1}{N^a}
    \sum_{y=0}^{N-1}   \,
    \mathrm{tr}
    \left( 
     U_y \, \rho_{E|0}^a \, U_y^\dag \; \rho_E^{1-a} 
    \right )
    \right\}\label{l2}\\
& =
    \frac{1}{1-a}
    \log \left\{
    \frac{1}{N^a}
    \sum_{y=0}^{N-1}   \,
    \mathrm{tr}
    \left( 
     \rho_{E|0}^a \, U_y^\dag \; \rho_E^{1-a} \, U_y \,  
    \right )
    \right\}\label{l3}\\
&=
        \frac{1}{1-a}
    \log \left\{
    \frac{N}{N^a} \,
    \mathrm{tr}
    \left( 
     \rho_{E|0}^a \;\rho_E^{1-a} 
    \right )
    \right\}\label{l4}\\  
    &=
    \log N + \frac{1}{1-a}
    \log \left\{
    \mathrm{tr}
    \left( 
     \rho_{E|0}^a \;\rho_E^{1-a} 
    \right )
    \right\} \,.
\end{align}
Here in line \eqref{l2} we write the different Eve's conditional states $\rho_{E|y}$, with $y=1,2,3$, in terms of $\rho_{E|0}$ using \eqref{symmetry}, then we observe that if $U$ is a unitary operator $(U\rho \,U^\dagger) ^a = U\rho^a \,U^\dagger$. 
(This latter property holds not only for unitary transformations, but also for isometries, as we show in the Lemma~\ref{isoexp} below.)
Finally, line \ref{l4} stems from invariance of Eve's state $\rho_E$ under the symmetry in Eq.~\eqref{symmetry}.

\begin{lemma}\label{isoexp}
    An isometry $V$ on a density matrix $\rho$ commutes with its exponentiation, \textit{i.e.}~for any $a\in \mathbb R$ we have
\begin{align}
    V \rho^a V^\dagger = (V \rho V^\dagger)^a \, .
\end{align}
\end{lemma}
\begin{proof}
Let the spectral decomposition of the density operator be
$
\rho = \sum_{j} \lambda_j \, |j\rangle\!\langle j|
$, with $\lambda_j \ge 0 $.
The action of the isometry on the basis vectors
$| \psi_j \rangle := V\,|j\rangle$.
preserve of orthonormality, indeed  since \(V^\dagger V = \mathbb I\), we can observe
$
\langle \psi_j | \psi_k \rangle
= \langle j | V^\dagger V | k \rangle
= \langle j | k \rangle
= \delta_{jk}
$.
Now the spectral form of the image can be written using the above vectors:
\begin{equation}
V \rho V^\dagger
= V \Bigl( \sum_j \lambda_j |j\rangle\!\langle j| \Bigr) V^\dagger
= \sum_j \lambda_j \, | \psi_j \rangle\!\langle \psi_j | \, ,
\end{equation}
and clearly, the eigenvalues remain unchanged. Then being $V \rho V^\dagger$ diagonal in $\{\ket{\psi_j\}}$, $(V \rho V^\dagger)^a$ will be diagonal as well, with exponentiated eigenvalues. Thus, for any real exponent $a$ we can  write
\begin{equation}
\bigl( V \rho V^\dagger \bigr)^{\!a}
= \sum_j \lambda_j^{a} \, | \psi_j \rangle\!\langle \psi_j |
= V\Bigl( \sum_j \lambda_j^{a} |j\rangle\!\langle j| \Bigr)V^\dagger
= V \rho^{a} V^\dagger .
\end{equation}
Hence \( V \rho^a V^\dagger = (V \rho V^\dagger)^a \). Actually, this argument holds not only for exponentiation, but for any Borel function $f$,
$
f(V \rho V) = V f(\rho) V
$.
\end{proof}

As for the optimized Petz-Rényi conditional entropy \eqref{PetzUp}, we apply the characterization in Lemma 1 of Ref.~\cite{BertaJMP} and obtain
\begin{align}
       H_a^\uparrow&(Y|E) 
     = \frac{a}{1-a}
    \log
    \left\{
    \mathrm{tr}
    \left[
    \mathrm{tr}_{Y}(\rho_{YE}^a)
    \right]^\frac{1}{a}
    \right\} \nonumber\\
& = \frac{a}{1-a}
    \log
    \left\{
    \mathrm{tr}
    \left[
    \mathrm{tr}_{Y} \left(
    \frac{1}{N^a}
    \sum_{y=0}^{N-1}   \,
    |y\rangle_Y\langle y| \otimes \rho_{E|y}^a
    \right)
    \right]^\frac{1}{a}
    \right\} \nonumber\\
& = \frac{a}{1-a}
    \log
    \left\{
    \mathrm{tr}
    \left[
    \frac{1}{N^a}
    \sum_{y=0}^{N-1}   \,
    \rho_{E|y}^a
    \right]^\frac{1}{a}
    \right\}
    \label{PetzUpl3}\\
    & = \frac{a}{1-a}
    \log
    \left\{
    \frac{1}{N}\,
    \mathrm{tr}
    \left[
    \sum_{y=0}^{N-1}   \,
    \rho_{E|y}^a
    \right]^\frac{1}{a}
    \right\}\label{PetzUpl4}\\
    & = -\frac{a}{1-a}\log N+
    \frac{a}{1-a}
    \log
    \left\{
    \mathrm{tr}
    \left[
    \sum_{y=0}^{N-1}   \,
    \rho_{E|y}^a
    \right]^\frac{1}{a}
    \right\} \, .
    \label{PetzUpl5}
\end{align}
In \eqref{PetzUpl3} we simply compute the partial trace over the classical register $Y$; the linearity of the trace is used in line \eqref{PetzUpl4} to take out the factor $1/N$, while \eqref{PetzUpl5} comes from the properties of the logarithm.

About the sandwiched Rényi conditional entropy \eqref{SandDown}, we obtain
\begin{align}
\tilde{H}_a^\downarrow& (Y|E)
 =\frac{1}{1-a} \,
    \log{
    \mathrm{tr} \left[
\left(
\rho_{E}^{\frac{1-a}{2 a}}
\rho_{YE}\;
\rho_{E}^{\frac{1-a}{2a}}
\right)^a
    \right]
    } \nonumber\\
=& \frac{1}{1-a} \,
    \log{
    \mathrm{tr} \left[
\left(
\rho_{E}^{\frac{1-a}{2 a}}
\left( 
\sum_{y=0}^{N-1}
 \frac{|y\rangle_Y\langle y | \otimes \rho_{E|y}}{N}\right) \;
\rho_{E}^{\frac{1-a}{2 a}}
\right)^a
    \right]
    }\nonumber \\
=&\frac{1}{1-a} \,
    \log{
    \mathrm{tr} \left[
\left( \frac{1}{N}\sum_{y=0}^{N-1}
 |y\rangle_Y\langle y | \otimes 
\rho_{E}^{\frac{1-a}{2 a}}
\rho_{E|y}\;
\rho_{E}^{\frac{1-a}{2 a}}
\right)^a
    \right]
    }\label{sandownl3} \\
=&\frac{1}{1-a} \,
    \log{
    \mathrm{tr} \left[
 \frac{1}{N^a}\sum_{y=0}^{N-1}
 |y\rangle_Y\langle y |  \otimes 
\left(\rho_{E}^{\frac{1-a}{2 a}}
\rho_{E|y}\;
\rho_{E}^{\frac{1-a}{2a}}
\right)^a
    \right]
    } \label{sandownl4}\\
=&\frac{1}{1-a} \,
    \log{
    \mathrm{tr} \left[
 \frac{1}{N^a}\sum_{y=0}^{N-1} 
\left(\rho_{E}^{\frac{1-a}{2 a}}
\rho_{E|y}\;
\rho_{E}^{\frac{1-a}{2a}}
\right)^a
    \right]
    } \label{sandownl5}\\
=&\frac{1}{1-a} \,
    \log{
     \frac{1}{N^a}\sum_{y=0}^{N-1}
    \mathrm{tr} \left[
\left(\rho_{E}^{\frac{1-a}{2 a}}
\rho_{E|y}\;
\rho_{E}^{\frac{1-a}{2a}}
\right)^a
    \right]
    } \nonumber
\end{align}

Here in line \eqref{sandownl3} we factorize operators on the same quantum system in the tensor product. Then in  \eqref{sandownl4} we use again the argument about  power of block-diagonal matrices. Finally in \eqref{sandownl5} we compute the partial trace on the classical register $Y$ using $\mathrm{tr}_{YE} = \mathrm{tr}_Y\mathrm{tr}_E$. 

Finally we consider the optimized sandwiched Rényi conditional entropy \eqref{SandUp}. We have:
\begin{align}
\tilde{H}_a^\uparrow (Y|E)
& = \sup_{\sigma_E}
\frac{1}{1-a} \,
    \log
     \sum_{y=0}^{N-1}    
    \mathrm{tr} \left[
\left(
\frac{
\sigma_{E}^{\frac{1-a}{2 a}}
\rho_{E|y}\;
\sigma_{E}^{\frac{1-a}{2a}}}{N}
\right)^a
    \right]\, .
\end{align}

A lower bound on this quantity is obtained by restricting the optimization to invariant states of the form
\begin{align}
    \sigma_E = U_t \, \sigma_E \, U_t^\dag \, .
\end{align}
Then, applying the same symmetry argument as above we obtain
\begin{align}
\tilde{H}_a^\uparrow (Y|E)
& \geq \log N  \,\nonumber\\
+ \frac{1}{1-a} 
   & \log
    \left(
     \inf_{\sigma \in \mathfrak{I}} 
    \mathrm{tr} \left[
\left(\sigma_{E}^{\frac{1-a}{2 a}}
\rho_{E|0}\;
\sigma_{E}^{\frac{1-a}{2a}}
\right)^a
    \right]
    \right) 
    \label{HsupAPP} \, ,
\end{align}
where $\mathfrak{I}$ denotes the set of invariant states.
(The infimum in \eqref{HsupAPP} appears because we are interested in $a >1$.)

\section{Derivation of CQ state $\rho_{YE}$ for passive attacks on NPSK protocols}
\label{cqpassive}
Here we explicitly obtain the CQ state $\rho_{YE}$ for NPSK protocol.
Alice holds a virtual computational state $\ket{\phi_x}$ corresponding to her preparation choice $x$. In this scheme, a generic state of the tripartite system $ABE$ of Alice, Bob and Eve after the transmission can be represented as
\begin{align}
    \ket{\Theta}_{ABE} = 
    \frac{1}{\sqrt{N}}
    \sum_{x=0}^{N-1}
    \ket{\phi_x}_A
    \ket{\sqrt{\eta}\;\alpha_x}_B 
    \ket{\sqrt{1-\eta}\;\alpha_x}_E
    \, .
\end{align}
The trace on Eve's system, $\rho_{AB}=\mathrm{tr}_E \rho_{ABE}$, where $\rho_{ABE}=\ket{\Theta}\bra{\Theta}_{ABE}$ yields
\begin{align}
    \rho_{AB} = \frac{1}{N} \sum_{x,x'=0}^{N-1} 
\ket{\phi_x}_A\bra{\phi_{x'}} \otimes 
\ket{\sqrt{\eta}~\alpha_x}_B\bra{\sqrt{\eta}~\alpha_{x'}} 
\nonumber\\
\times \exp\left[-|\alpha|^2(1-\eta)
\left( 
1 - e^{-i(x'-x)\frac{\pi}{2}} 
\right)\right]
\, .
\label{rhoABstate}
\end{align}
As usual in QKD, Eve's attack is formalized by the action of the \textit{key map} $\mathcal{G}:S(ABE)\mapsto S(ABY_1E)$ (see Eq.~\eqref{gkeymap})
\begin{align}
    \rho_{ABY_1E} &= \mathcal G (\rho_{ABE})
    =\frac{1}{N}
    \sum_{y,y'=0}^{N-1}
    \sum_{x,x'=0}^{N-1}
    \ket{\phi_x}_A\bra{\phi_{x'}}\nonumber\\
    &\otimes
    \sqrt{\Lambda_y}\ket{\sqrt{\eta}\;\alpha_x}_B
    \bra{\sqrt{\eta}\;\alpha_{x'}}\sqrt{\Lambda_{y'}}
    \nonumber\\
    &\otimes \ket{y}_{Y_1}\bra{y'}
    \otimes\ket{\sqrt{1-\eta}\;\alpha_x}_E\bra{\sqrt{1-\eta}\;\alpha_{x'}}
     \, .
\end{align}
followed by an isometry (see Eq.~\eqref{isom}),
\begin{align}
    &\rho_{ABYY_1E} = V\rho_{ABEY_1}V^\dagger 
    =\frac{1}{N}
    \sum_{y,y'=0}^{N-1}
    \sum_{x,x'=0}^{N-1}
    \left(\ket{\phi_x}\bra{\phi_{x'}}\right)_A\nonumber\\
    &\quad\otimes
    \left(\sqrt{\Lambda_y}\ket{\sqrt{\eta}\;\alpha_x}
    \bra{\sqrt{\eta}\;\alpha_{x'}}\sqrt{\Lambda_{y'}}\right)_B
    \otimes 
    \left(\ket{y}\bra{y'}\right)_{Y}\nonumber\\
    &\quad
    \otimes 
    \left(\ket{y}\bra{y'}\right)_{Y_1}
    \otimes
    \left(\ket{\sqrt{1-\eta}\;\alpha_x}\bra{\sqrt{1-\eta}\;\alpha_{x'}}\right)_E\,.
     \, \label{eq:cqstate}
\end{align}
As discussed in detail in Appendix \ref{dual}, we are interested in how much knowledge about $Y$ is contained inside Eve's quantum system $E$, then Eq.~\eqref{CQdef} specializes in
\begin{equation}
   \rho_{YE}= \mathrm{tr}_{ABY_1} (\rho_{ABYY_1E})=\frac1N
    \sum_{y=0}^{N-1}
    \ket{y}_Y\bra{y}\otimes
    \rho_{E|y}
\end{equation}
where 
\begin{align}
    \rho_{E|y} 
    = \sum_{x=0}^{N-1}
    \braket{\Lambda_y}_x
    \ket{\sqrt{1-\eta}\;\alpha_x}_E\bra{\sqrt{1-\eta}\;\alpha_{x}}
\end{align}
is Eve's conditional state with
$
     \braket{\Lambda_y}_x 
     =~_{B}\braket{\sqrt{\eta}\,\alpha_x
    |\Lambda_y|
    \sqrt{\eta}\,\alpha_x}_B 
$. 
In particular, and this is what matters for us, the reduced state $\rho_{YE}$ has the structure of a CQ state, where $Y$ represents the outcome of Bob's measurement and is therefore treated as classical information. 

\section{Duality relations}
\label{dual}
In this paper, we computed different bounds of the secret-key rate for BPSK and QPSK protocols, under the assumption of a specific and known attack performed by the eavesdropper. This allowed us to know analytically the explicit form of the states describing Eve's quantum system $E$ on which the various entropic functionals are computed. This particular case study can indeed be a very interesting starting point, providing some benchmarks on the values of the secret-key rates, especially in the finite-size regime. A more complete analysis will require considering a generic and unknown attack performed by the eavesdropper. In this case, we cannot know the form of the final state in Eve's possession; therefore, we are not able to compute directly the different entropies like we did throughout this paper. 
In this scenario, \textit{duality relations} can be employed to estimate quantities that depend on Eve's system in terms of other ones that are only functions of Alice and Bob's quantum states.
Let $\rho_{ABC}$ a tripartite pure state; we can consider three different duality relations involving Rényi entropies \cite{BertaIEEE}
\begin{align}
     &\left[H_a^\downarrow(A|B) + H_b^\downarrow(A|C)\right]_{\substack{a,b\in[0,2] \\a+b=2}} = 0,
     \label{dual1}\\
     & \left[\tilde{H}_a^\uparrow(A|B) +\tilde{H}_b^\uparrow(A|C)\right]_{\substack{a,b\ge1/2\\a^{-1}+b^{-1}=2}}=0,
     \label{dual2}\\
     & [{H}_a^\uparrow(A|B) +\tilde{H}_b^\downarrow(A|C)]_{\substack{a,b>0\\a\,b=1}}=0\,.
     \label{dual3}
\end{align}
In our framework, we start considering a tripartite system $ABE$, then to  mathematically describe Bob's measurements, it is possible to extend it to $ABYY_1E$, where both $Y$ and $Y_1$ are classical registers.  First of all we introduce a \textit{key map} $\mathcal G$ \cite{reliableColes} such that
\begin{align}\label{gkeymap}
\mathcal{G}_{B \to Y_1 B} = \sum_{y=0}^{N-1} |y\rangle_{Y_1} \otimes \sqrt{\Lambda_y}   
 \, .
\end{align}
Here $\Lambda_y\equiv \Lambda_{y(\gamma)}$ refers to a region of the optical phase $R_y$ space integrating over the coherent state $\ket{\gamma}$ that assigns a discrete value $y$. Clearly $\mathcal{G}$ acts trivially on $A$ and $E$. The role of this map is to store Bob's measurement outcomes into an orthonormal basis $\{\ket{y}_Y\}$, $y=0,...,N-1$.Once Bob stored the value in the classical register $Y_1$, the attack of Eve can be modelled by the following isometry that copy the information from the register $Y_1$ to the register $Y$.
\begin{align}
V_{Y_1 \to Y Y_1} = \sum_{y=0}^{N-1} |y\rangle_{Y} 
\otimes | y \rangle_{Y_1} \langle y|,
 \qquad V^\dag V=\mathbb I.
 \label{isom}
\end{align}
This defines the projector $\Pi=VV^\dag $ on the subspace with $Y=Y_1$. Note that $Y_1$ simply holds a copy of $Y$ in the computational basis and that if we trace out system $Y_1$, the resulting state is CQ, as we obtained in the specialized case of passive attacks on NPSK protocols in Appendix \ref{cqpassive}.  Hence, the state we are interested in is the pure state $\rho_{ABYY_1E} = V\mathcal G(\rho_{ABE})V^\dagger$ over five quantum systems. 

\subsection{Duality between Petz Rényi entropies \eqref{dual1}}

Consider  $\rho_{ABYY_1E} = V\mathcal G(\rho_{ABE})V^\dagger$, then Eq.~\eqref{dual1} reads
\begin{align}
    H_a^\downarrow( Y |E) + H_b^\downarrow( Y |A B Y_1 ) = 0, \quad a,b\in[0,2],  \, a+b =2\,.
\end{align}
Therefore, here we should focus on the quantity
\begin{align}
H_b^\downarrow(Y|ABY_1) &= 
\frac{1}{1-b} \,
\log
\mathrm{tr} \left\{
 \rho_{YABY_1}^b \,
\left(\mathbb I_{Y}\otimes\rho_{ABY_1}^{1-b}
\right)
\right\} \, .
\end{align}

Note that the state $\rho_{ABY_1}$ is CQ:
\begin{align}
    \rho_{ABY_1}
    =
    \mathcal{Z}(\mathcal{G}(\rho_{AB})) \, ,
    \label{ZG}
\end{align}
where $\mathcal{Z}$ is the well-known \textit{pinching map}~\cite{reliableWinick}, which induces a complete dephasing in the basis of the classical register $Y_1$. We use the fact that
\begin{align}
     \rho_{ABYY_1} = V \mathcal{G}(\rho_{AB}) V^\dagger \, , 
\end{align}
together with Lemma \eqref{isoexp} of Appendix~\ref{sec2}:
\begin{align}
     \left( 
     V \mathcal{G}(\rho_{AB}) V^\dagger
     \right)^b
     =
     V
     \left( 
     \mathcal{G}(\rho_{AB}) 
     \right)^b
     V^\dagger \, .
\end{align}
We thus obtain
\begin{align}
H_b^\downarrow(Y|ABY_1) 
& \\
=\frac{1}{1-b} \,
&\log \,
\mathrm{tr} \left\{
 \left( 
     V \mathcal{G}(\rho_{AB}) V^\dag
     \right)^b \,
\left(\mathbb I_{Y}\otimes\rho_{ABY_1}^{1-b}
\right)
\right\} \nonumber\\
= 
\frac{1}{1-b} \,
&\log \,
\mathrm{tr} \left\{
 V 
     \mathcal{G}(\rho_{AB})^b V^\dag \,
\left(\mathbb I_{Y}\otimes\rho_{ABY_1}^{1-b}
\right)
\right\} \nonumber\\
= 
\frac{1}{1-b} \,
&\log \,
\mathrm{tr} \left\{
 \mathcal{G}(\rho_{AB})^b 
 \,    V^\dag 
\left(\mathbb I_{Y}\otimes\rho_{ABY_1}^{1-b} 
\right) V
\right\}  \, ,\nonumber
\end{align}
where in the last equality we used the ciclity property of the trace. Now we use the following identity:

\begin{align}
V^\dag \left(
\mathbb I_{Y} \otimes \sigma_{ABY_1} 
\right)
V
 &=\left( \sum_{y} {}_{Y}\langle y | \otimes 
|y\rangle_{Y_1}\langle y| 
\right)\nonumber\\
\times&\left(
\mathbb I_{Y}\otimes \sigma_{ABY_1} 
\right)
\left( \sum_{y'} |y'\rangle_{Y} \otimes |y'\rangle_{Y_1}\langle y'| 
\right) \nonumber\\
& =
\sum_{y} 
|y\rangle_{Y_1} \langle y| \otimes
\langle y| \sigma_{ABY_1} | y \rangle \, .
\end{align}
Note that this is an implementation of the pinching map. Using Eq.~\eqref{ZG} with $\mathcal{Z}^2 = \mathcal{Z}$,
we finally obtain
\begin{align}
H_b^\downarrow(Y|ABY_1) 
&= 
\frac{1}{1-b} \,
\log
\mathrm{tr} \left\{
     \mathcal{G}(\rho_{AB})^b 
\mathcal{Z}(\mathcal{G}(\rho_{AB}))
^{1-b}
     \right\}  \, .
\end{align}
We can make this more explicit as

\begin{align}
&H_b^\downarrow(Y|ABY_1) 
 = 
\frac{1}{1-b} \,
\log
\mathrm{tr} \biggl\{\nonumber\\
     &\mathcal{G}(\rho_{AB})^b 
\left(
\sum_y |y\rangle \langle y| \otimes 
\left( \sqrt{\Lambda_y} \, \rho_{AB}
\sqrt{\Lambda_y} \right)^{1-b}
\right)
     \biggr\} \nonumber\\
& = 
\frac{1}{1-b} \,
\log
\mathrm{tr} \left\{
     \mathcal{Z}( \mathcal{G}(\rho_{AB})^b )
\mathcal{Z}(\mathcal{G}(\rho_{AB}))
^{1-b}
     \right\} \nonumber\\
& = 
\frac{1}{1-b} \,
\log
\sum_y  \mathrm{tr} \left\{
     \langle y| \mathcal{G}(\rho_{AB})^b 
     | y \rangle
\left( \sqrt{\Lambda_y} \, \rho_{AB}
\sqrt{\Lambda_y} \right)^{1-b}
     \right\} \, .
\end{align}

The expression may further simplify if $\mathcal{G}$ is also an isometry (\textit{i.e.}~no use of post-selection). In this case we use the identity $\mathcal{G}(\sigma)^b=\mathcal{G}(\sigma^b) \, .$ (that can be proven again with Lemma \ref{isoexp} in appendix \ref{sec2}). This yields

\begin{align}
&H_b^\downarrow(Y|ABY_1) 
= 
\frac{1}{1-b} \,
\log
\mathrm{tr} \left\{
     \mathcal{G}(\rho_{AB}^b) 
\mathcal{Z}(\mathcal{G}(\rho_{AB}))
^{1-b}
     \right\}  \nonumber\\
&= 
\frac{1}{1-b} \,
\log
\mathrm{tr} \left\{
     \rho_{AB}^b \,
\mathcal{G}^\dag(\mathcal{Z}(\mathcal{G}(\rho_{AB}))
^{1-b})
     \right\}\nonumber  \\
&= 
\frac{1}{1-b} \,
\log
\mathrm{tr} \sum_y 
\left\{
   \sqrt{\Lambda_y} \,  \rho_{AB}^b 
\sqrt{\Lambda_y} \left( \sqrt{\Lambda_y} \, \rho_{AB}
\sqrt{\Lambda_y} \right)^{1-b} 
     \right\} 
\end{align}
where $\mathcal{G}^\dag$ is the dual. Thus the bound \eqref{RateSand} on the min-entropy becomes
\begin{align}
     \tilde H_\text{min}^{\uparrow\epsilon}&(Y^n|E^n)
     \geq   H^\downarrow_a(Y^n|E^n) - \frac{g(\epsilon)}{a -1}\nonumber\\
    &=-n  {H}_{2-a}^\downarrow(Y|ABY_1) - \frac{g(\epsilon)}{a -1}
    \, ,
    \qquad a\in(1,2].
\end{align} 

\subsection{Duality between optimized sandwiched Rényi entropies \eqref{dual2}}
Consider  $\rho_{ABYY_1E} = V\mathcal G(\rho_{ABE})V^\dagger$, then Eq.~\eqref{dual2} reads
\begin{align}
    \tilde{H}_a^\uparrow(Y |E ) +\tilde{H}_b^\uparrow(Y |A B Y_1 )=0, \nonumber\\
\end{align}
with $a^{-1}+b^{-1}=2$, and $a,b \ge\frac{1}{2}.$
Therefore, here we should focus on the quantity 
\begin{widetext}
\begin{align}
    \tilde{H}_b^\uparrow(Y|AB Y_1) 
    & = \sup_{\sigma_{AB Y_1}}
    \frac{1}{1-b} \,
    \log{
    \mathrm{tr} \left[
\left(
\mathbb{I}_Y \otimes \sigma_{AB Y_1}^{\frac{1-b}{2b}}
\, \rho_{AB YY_1} \,
\mathbb{I}_Y \otimes \sigma_{AB Y_1}^{\frac{1-b}{2b}}
\right)^b
    \right]
    } \nonumber\\
& = \sup_{\sigma_{AB Y_1}}
    \frac{1}{1-b} \,
    \log{
    \mathrm{tr} \left[
\left(
\mathbb{I}_Y \otimes \sigma_{AB Y_1}^{\frac{1-b}{2b}}
\, V \mathcal{G}(\rho_{AB}) V^\dag \,
\mathbb{I}_Y \otimes \sigma_{AB Y_1}^{\frac{1-b}{2b}}
\right)^b
    \right]
    } \, .
\end{align}
We can use the alternate characterization of this entropy:
\begin{align}
\tilde{H}_b^\uparrow(Y|ABY_1) 
& = \sup_{\sigma_{ABY_1}} 
\frac{b}{1-b}
\log
\left\| (V \mathcal{G}(\rho_{AB}) V^\dag)^{\frac12} \,
\mathbb{I}_Y\otimes \sigma_{ABY_1}^{\frac{1-b}{b}} \, 
(V \mathcal{G}(\rho_{AB}) V^\dag)^{\frac12}
\right\|_b 
\label{eq40}\nonumber\\
& = \sup_{\sigma_{ABY_1}} 
\frac{b}{1-b}
\log
\left\| V \mathcal{G}(\rho_{AB})^{\frac12} V^\dag
\mathbb{I}_Y\otimes \sigma_{ABY_1}^{\frac{1-b}{b}} V \mathcal{G}(\rho_{AB})^{\frac12}
V^\dag
\right\|_b \nonumber\\
& = \sup_{\sigma_{ABY_1}} 
\frac{b}{1-b}
\log
\left\| \mathcal{G}(\rho_{AB})^{\frac12}
\mathcal{Z}( \sigma_{ABY_1}^{\frac{1-b}{b}} ) \mathcal{G}(\rho_{AB})^{\frac12}
\right\|_b \nonumber\\
& = \sup_{\sigma_{AB Y_1}}
    \frac{1}{1-b} \,
    \log{
    \mathrm{tr} \left[
\left(
\mathcal{Z}( \sigma_{ABY_1}^{\frac{1-b}{b}} )^\frac12
\, \mathcal{G}(\rho_{AB}) \,
\mathcal{Z}( \sigma_{ABY_1}^{\frac{1-b}{b}} )^\frac12
\right)^b
    \right]
    } \, .
\end{align}
Thus the bound \eqref{RateSand} on the min-entropy becomes
\begin{align}
     \tilde H_\text{min}^{\uparrow\epsilon}(Y^n|E^n)
     \geq  \tilde H^\uparrow_a(Y^n|E^n) - \frac{g(\epsilon)}{a -1}
    =-n \tilde {H}_{\frac{a}{2a-1}}^\uparrow(Y|ABY_1) - \frac{g(\epsilon)}{a -1} \, ,
    \qquad a>1.
\end{align} 
\subsection{Duality between the sandwiched Rényi entropy and the optimized Petz Rényi entropy \eqref{dual3}}
Consider  $\rho_{ABYY_1E} = V\mathcal G(\rho_{ABE})V^\dagger$, then Eq.~\eqref{dual3} can be exploited in two different ways.

Case 1:
\begin{align} 
    \tilde{H}_a^\downarrow(Y |E ) +{H}_b^\uparrow(Y |A B Y_1 ) =0,  \qquad a,b >0, \qquad a\cdot b=1.
\end{align}

Here we should focus on the quantity 
\begin{align}
H_b^\uparrow(Y|ABY_1) &= 
\sup_{\sigma_{ABY_1}}
\frac{1}{1-b} \,
\log
\mathrm{tr} \left\{
\left[
 \rho_{YABY_1}^b \,
\left(\mathbb I_{Y}\otimes\sigma_{ABY_1}^{1-b}
\right)
\right]
\right\} \, .
\label{eq49}
\end{align}
We can try to get rid of the supremum following \cite{BertaIEEE}, using the Holder's inequality
\begin{align}
    \mathrm{tr}[AB] \le  \mathrm{tr}[A^p]^{\frac{1}{p}} \mathrm{tr}[B^q]^{\frac{1}{q}}, \qquad \frac{1}{p}+\frac{1}{q} = 1
\end{align}
From Eq.~(\ref{eq49}) we can write
\begin{align}
    H_b^\uparrow(Y|ABY_1) &= 
\sup_{\sigma_{ABY_1}}
\frac{1}{1-b} \,
\log
\mathrm{tr} \left\{
\left[
 \mathrm{tr_{Y}}(\rho_{YABY_1}^b) \,
\sigma_{ABY_1}^{1-b}
\right]
\right\}
\end{align}
and identifying $A=\mathrm{tr_{Y}}(\rho_{YABY_1}^b)$, $B=\sigma_{ABY_1}^{1-b}$, $p=\frac{1}{b}$, $q=\frac{1}{1-b}$
\begin{align}
    H_b^\uparrow(Y|ABY_1)&\le
    \sup_{\sigma_{ABY_1}}\frac{1}{1-b}
    \log
    \left\{
    \left[
    \mathrm{tr}
    \left[
    \mathrm{tr}_{Y}(\rho_{B_1ABB_2Y_1}^b)
    \right]^\frac{1}{b}
    \right]^b\cdot
    \mathrm{tr}[(\sigma_{ABY_1})] ^\frac{1}{1-b}
    \right\}\nonumber\\
    &=\frac{1}{1-b}
    \log
    \left\{
    \left[
    \mathrm{tr}
    \left[
    \mathrm{tr}_{Y}(\rho_{YABY_1}^b)
    \right]^\frac{1}{b}
    \right]^b
    \right\}\nonumber\\
    &=\frac{b}{1-b}
    \log
    \left\{
    \mathrm{tr}
    \left[
    \mathrm{tr}_{Y}(\rho_{YABY_1}^b)
    \right]^\frac{1}{b}
    \right\} \, .
\end{align}
This quantity happens to be trivially also a lower bound for $ H_b^\uparrow(Y|ABY_1)$, by definition of supremum, if we chose a particular realization of $\sigma^*_{AB\textbf{}}$:
\begin{align}
\sigma_{ABY_1}^* = \frac{\mathrm{tr}_{Y}(\rho_{YABY_1}^b)
    ^\frac{1}{b}}{\mathrm{tr}\left[\mathrm{tr}_{Y}(\rho_{Y AB Y_1}^b)
    ^\frac{1}{b}\right]}
\end{align}
Thus we can claim
\begin{align}
      H_b^\uparrow(Y|ABY_1) = \frac{b}{1-b}
    \log
    \left\{
    \mathrm{tr}
    \left[
    \mathrm{tr}_{Y}(\rho_{YABY_1}^b)
    \right]^\frac{1}{b}
    \right\} \, .
\end{align}

Now recall that $\rho_{ABYY_1} = V \mathcal{G}( \rho_{AB} ) V^\dagger$.
Then we obtain
\begin{align}
      H_b^\uparrow(Y|ABY_1) 
      & = \frac{b}{1-b}
    \log
    \left\{
    \mathrm{tr}
    \left[
    \mathrm{tr}_{Y}( (V \mathcal{G}( \rho_{AB} ) V^\dag)^b )
    \right]^\frac{1}{b}
    \right\} \nonumber\\
& = \frac{b}{1-b}
    \log
    \left\{
    \mathrm{tr}
    \left[
    \mathrm{tr}_{Y}( V \mathcal{G}( \rho_{AB} )^b V^\dag )
    \right]^\frac{1}{b}
    \right\} \nonumber \\
& = \frac{b}{1-b}
    \log
    \left\{
    \mathrm{tr}
    \left[
    \mathrm{tr}_{Y}( \mathcal{G}( \rho_{AB} )^b  )
    \right]^\frac{1}{b}
    \right\}  \, .
\end{align}

More explicitly, 
\begin{align}
      H_b^\uparrow(Y|ABY_1) 
& = \frac{b}{1-b}
    \log
    \left\{
    \mathrm{tr}
    \left[
    \sum_y \langle y | 
    \mathcal{G}( \rho_{AB} )^b | y \rangle 
    \right]^\frac{1}{b}
    \right\} \, . 
\end{align}

Thus the bound \eqref{RateSand} on the min-entropy becomes
\begin{align}
     \tilde H_\text{min}^{\uparrow\epsilon}(Y^n|E^n)
     \geq \tilde H^\downarrow_a(Y^n|E^n) - \frac{g(\epsilon)}{a -1}
    =-n {H}_{\frac{1}{a}}^\uparrow(Y|ABY_1) - \frac{g(\epsilon)}{a -1} \, ,
    \qquad a\in(1,2].
\end{align} 

Case 2:
\begin{align} 
    {H}_a^\uparrow(Y |E ) +\tilde{H}_b^\downarrow(Y |A B Y_1 )=0,  \qquad a,b >0, \qquad a\cdot b=1.
\end{align}
Here we should focus on the quantity 
\begin{align}
    \tilde{H}_b^\downarrow(Y |A B Y_1 ) 
    & = 
    \frac{1}{1-b}
    \log
    \left \| 
    (\rho_{Y AB Y_1})^\frac12 \,
    (\mathbb I_{Y} \otimes\rho_{AB Y_1})^\frac{1-b}{b} \,
    (\rho_{Y AB Y_1})^\frac12 \,
    \right \|_b \nonumber\\
& = 
    \frac{1}{1-b}
    \log
    \left \| 
    ( V \mathcal{G}(\rho_{AB}) V^\dag )^\frac12 \,
    (\mathbb I_{Y} \otimes\rho_{AB Y_1})^\frac{1-b}{b} \,
    ( V \mathcal{G}(\rho_{AB}) V^\dag )^\frac12 \,
    \right \|_b\nonumber\\
& = 
    \frac{1}{1-b}
    \log
    \left \| 
    \mathcal{G}(\rho_{AB})^\frac12 
    V^\dag
    (\mathbb I_{Y} \otimes \rho_{AB Y_1})^\frac{1-b}{b} \,
    V \mathcal{G}(\rho_{AB})^\frac12 \,
    \right \|_b \, .
\end{align}
Recall that $\rho_{AB Y_1} = \mathcal{Z} (\mathcal{G} (\rho_{AB}))$, which implies
\begin{align}
    \tilde{H}_b^\downarrow(Y |A B Y_1 ) 
    & = 
    \frac{1}{1-b}
    \log
    \left \| 
    \mathcal{G}(\rho_{AB})^\frac12 
    \mathcal{Z} (\mathcal{G} (\rho_{AB}))^\frac{1-b}{b} \,
    \mathcal{G}(\rho_{AB})^\frac12 \,
    \right \|_b \nonumber\\
& = \frac{1}{1-b} \,
    \log{
    \mathrm{tr} \left[
\left(
\mathcal{Z} (\mathcal{G} (\rho_{AB}))^\frac{1-b}{b}
\, \mathcal{G}(\rho_{AB}) \,
\mathcal{Z} (\mathcal{G} (\rho_{AB}))^\frac{1-b}{b} )
\right)^b
    \right]
    }  \, .
\end{align}

Thus the bound \eqref{RateSand} on the min-entropy becomes
\begin{align}
     &\tilde H_\text{min}^{\uparrow\epsilon}(Y^n|E^n)
    \geq  H^\uparrow_a(Y^n|E^n) - \frac{g(\epsilon)}{a -1}
    \nonumber\\
    &\quad=-n \tilde {H}_{\frac{1}{a}}^\downarrow(Y|ABY_1) - \frac{g(\epsilon)}{a -1} \,, 
    \qquad a\in(1,2].
\end{align} 
\end{widetext}

\section{Thermal noisy channel}
\label{sec:thermal}
For both BPSK ($N=2$) and QPSK ($N=4$), the usual EB representation is
\[
\ket{\Phi}_{AA'} = \sum_{x=0}^{N-1} \sqrt{p_x} \ket{x}_A \ket{\alpha_x}_{A'},
\]
with
$p_x = \frac{1}{N}$ uniform modulation), $
\alpha_x = (-1)^x \alpha$ for BPSK, 
$\alpha_k = i^k e^{i\pi/4} \alpha$ for QPSK.
So the classical register in the PM picture becomes a quantum system $A$ entangled with the signal mode $A'$.
Eve’s input state on mode $E_1$ is thermal with mean photon number $\mu$:
\begin{align}
\rho_{E_1} &= \frac{1}{\pi\mu} \int d^2\beta\, e^{-|\beta|^2/\mu} \ket{\beta}\bra{\beta}
\nonumber\\
&= \frac{1}{1+\mu} \sum_{n=0}^\infty \left( \frac{\mu}{1+\mu} \right)^n \ket{n}\bra{n}.    
\end{align}
A convenient purification is the usual two–mode squeezed vacuum:
\begin{align}
    \ket{\Psi}_{E_1E_2} = \sum_{n=0}^\infty c_n \ket{n}_{E_1} \ket{n}_{E_2}, \qquad 
c_n = \sqrt{ \frac{\mu^n}{(1+\mu)^{n+1}} }.
\end{align}
By construction,
\begin{align}
    \rho_{E_1} = \mathrm{tr}_{E_2} \left[ \ket{\Psi}\bra{\Psi}_{E_1E_2} \right].
\end{align}
Before the channel, the global pure state is
\begin{align}
\ket{\Phi_{\text{in}}}_{AA'E_1E_2}
&= \ket{\Phi}_{AA'} \otimes \ket{\Psi}_{E_1E_2}
\nonumber\\
&= \sum_{x=0}^{N-1} \sqrt{p_x} \ket{x}_A \ket{\alpha_x}_{A'} 
\otimes \sum_{n=0}^\infty c_n \ket{n}_{E_1} \ket{n}_{E_2}.
\end{align}
We now let a beam splitter of transmittance $\eta$ act only on modes $A'$ and $E_1$.
Let $a_{A'}$, $a_{E_1}$ be the annihilation operators of the input modes, and $b$, $e_1'$ for the output modes (Bob and Eve’s first mode). A standard beam splitter convention is:
\begin{align}
b &= \sqrt{\eta} a_{A'} + \sqrt{1-\eta} a_{E_1}, \\
e_1' &= -\sqrt{1-\eta} a_{A'} + \sqrt{\eta} a_{E_1}.
\end{align}
There exists a unitary $U_\eta^{A'E_1}$ such that:
\begin{align}
    b = U_\eta^{A'E_1} a_{A'} (U_\eta^{A'E_1})^\dagger, \qquad
e_1' = U_\eta^{A'E_1} a_{E_1} (U_\eta^{A'E_1})^\dagger.
\end{align}
The full unitary is:
\begin{align}
    U = \mathbb{I}_A \otimes U_\eta^{A'E_1} \otimes \mathbb{I}_{E_2}.
\end{align}
Write $\ket{\alpha_x}_{A'} = D_{A'}(\alpha_x) \ket{0}_{A'}$, where:
\begin{align}
    D_{A'}(\alpha_x) = \exp\left( \alpha_x a_{A'}^\dagger - \alpha_x^* a_{A'} \right).
\end{align}
Under the BS unitary:
\begin{align}
    U_\eta^{A'E_1} D_{A'}(\alpha_x) (U_\eta^{A'E_1})^\dagger
= D_B(\sqrt{\eta}\,\alpha_x) D_{E_1'}(-\sqrt{1-\eta}\,\alpha_x),
\end{align}
i.e., the displacement splits between Bob and Eve. The output state for a symbol $x$ becomes:
\begin{align}
    \ket{\Xi_x}_{BE_1'E_2} 
= D_B(\sqrt{\eta}\,\alpha_x) D_{E_1'}(-\sqrt{1-\eta}\,\alpha_x) \ket{\Xi_0}_{BE_1'E_2},
\end{align}
where

\begin{align}
    \ket{\Xi_0}_{BE_1'E_2}
= U_\eta^{A'E_1} \left( \ket{0}_{A'} \otimes \ket{\Psi}_{E_1E_2} \right)
\end{align}
is a fixed 3-mode Gaussian state.
Thus the global output state is:
\begin{align}
    \ket{\Phi_{\text{out}}}_{ABE_1'E_2}
= \sum_{x=0}^{N-1} \sqrt{p_x} \ket{x}_A \ket{\Xi_x}_{BE_1'E_2}.
\end{align}
From Bob’s point of view, the channel $A' \to B$ is thermal-loss with transmittance $\eta$ and environment mean photon number $\mu$, so that after the beam splitter:
\begin{align}
    \rho_{B|x} = \mathcal{E}_{\eta,\mu}\left( \ket{\alpha_x}\bra{\alpha_x} \right)
\end{align}
is a displaced thermal state with mean amplitude: $\sqrt{\eta} \alpha_x$  and variance
$
V_B = \frac{1}{2} + (1-\eta)\mu
$. 
This is the standard “entangling cloner” representation.
For BPSK, we compute $p(y|x)$. In the pure-loss case:
\begin{equation}
p(q|x) = \frac{1}{\sqrt{\pi}} \exp\left[ -\left( q - (-1)^x \sqrt{2\eta}|\alpha| \right)^2 \right],
\end{equation}
(mean: $(-1)^x \sqrt{2\eta}|\alpha|$, variance $1/2$). With thermal noise this transforms into:
\begin{align}
    p(q|x) = \frac{1}{\sqrt{2\pi V_B}} 
\exp\left[ -\frac{\left(q - (-1)^x \sqrt{2\eta}|\alpha| \right)^2}{2V_B} \right].
\end{align}
Bob’s decision rule:
\begin{align}
    y = \begin{cases}
0 & \text{if } q > 0, \\
1 & \text{if } q \leq 0.
\end{cases}
\end{align}
Hence the probability of correct detection is:
\begin{align}
\label{thermalhomo}
    &p(y = x | x) = \Phi\left( \frac{\sqrt{2\eta}|\alpha|}{\sqrt{V_B}} \right)\\
&= \frac{1}{2} \left[ 1 + \mathrm{erf}\left( \frac{\sqrt{\eta}|\alpha|}{\sqrt{\tfrac{1}{2} + (1-\eta)\mu}} \right) \right],
\end{align}
and
\begin{align}
    p(y \neq x | x) = 1 - p(y = x | x).
\end{align}
One can check that for $\mu=0$ Eq. \eqref{pydifx} is recovered:
\[
p(y = x | x) = \frac{1}{2} \left[ 1 + \mathrm{erf}(\sqrt{2\eta}|\alpha|) \right].
\]
For QPSK $p(y|k)$, in pure-loss heterodyne distribution:
\begin{align}
    p(\gamma|\alpha_k) = \frac{1}{\pi} \exp\left( -|\gamma - \sqrt{\eta} \alpha_k|^2 \right),
\end{align}
\begin{align}
    \mathrm{Re}\,\gamma&\sim \mathcal{N} \left( \mathrm{Re}(\sqrt{\eta} \alpha_k), \tfrac{1}{2} \right)\\
     \mathrm{Im}\,\gamma &\sim \mathcal{N} \left( \mathrm{Im}(\sqrt{\eta} \alpha_k), \tfrac{1}{2} \right).
\end{align}
Now with thermal noise:
\begin{align}
    p(\gamma|\alpha_k) 
= \frac{1}{\pi(1 + \bar{n}_B)} 
\exp\left( -\frac{|\gamma - \sqrt{\eta} \alpha_k|^2}{1 + \bar{n}_B} \right).
\end{align}
with $\bar{n}_B = (1 - \eta)\mu,$ and the
quadrature variance:
\begin{align}
    V_{\text{het}} = \frac{1 + (1 - \eta)\mu}{2}.
\end{align}
If you define:
\begin{align}
    P_\pm^{(\mu)} = \frac{1}{2} \left[ 1 \pm \mathrm{erf}\left( \frac{\sqrt{\eta}|\alpha|}{\sqrt{2[1 + (1 - \eta)\mu]}} \right) \right],
\end{align}
then
\begin{align}
\label{thermalhetero}
    p(y|k) =
\begin{cases}
(P_+^{(\mu)})^2 & y = k, \\[4pt]
(P_-^{(\mu)})^2 & |y-k| = 2, \\[4pt]
P_+^{(\mu)} P_-^{(\mu)} & |y-k| = 1 \text{ or } 3.
\end{cases}
\end{align}
For $\mu=0$ this reduces to the original $P_\pm$ in Eq. \eqref{eq:p_pm} and the matrix \eqref{matrix} in the paper.
So to generalize Eq. \eqref{CQ_qpsk} you just keep exactly the same structure but replace $P_\pm$ with $P_\pm^{(\mu)}$ above.
Now, we need to determined the Eve's conditional state $\rho_{E_1E_2|x}$ in the Fock space. Fix an input symbol $x$. The 3-mode output state on $(B,E_1',E_2)$ is
\begin{align}
\ket{\Xi_x}_{BE_1'E_2} = U_{\eta}^{A'E_1} \left( \ket{\alpha_x}_{A'} \otimes \ket{\Psi}_{E_1E_2} \right).
\end{align}
Eve’s conditional state before Bob’s measurement is
\begin{align}
\rho_{E_1'E_2|x} = \mathrm{tr}_B \left[ \ket{\Xi_x}\bra{\Xi_x}_{BE_1'E_2} \right].
\end{align}
This is the target state to represent in the Fock basis. The coherent state and the TMSV in the Fock basis are
\begin{align}
\ket{\alpha_x}_{A'} &= e^{-|\alpha_x|^2/2} \sum_{k=0}^\infty \frac{\alpha_x^k}{\sqrt{k!}} \ket{k}_{A'}, \\
\ket{\Psi}_{E_1E_2} &= \sum_{n=0}^\infty c_n \ket{n}_{E_1} \ket{n}_{E_2}, \quad 
c_n = \sqrt{ \frac{\mu^n}{(1+\mu)^{n+1}} }.
\end{align}
So the 3-mode input is:
\begin{align}
\ket{\text{in}_x}_{A'E_1E_2} &= e^{-|\alpha_x|^2/2} \sum_{k=0}^\infty \sum_{n=0}^\infty \frac{\alpha_x^k}{\sqrt{k!}} c_n \nonumber\\
&\qquad \times \ket{k}_{A'} \ket{n}_{E_1} \ket{n}_{E_2}.
\end{align}
The beam splitter action on $(A',E_1)$ gives:
\begin{align}
U_{\eta}^{A'E_1} \ket{k}_{A'} \ket{n}_{E_1} = \sum_{j=0}^{k+n} U_{k,n}^{(j)} \ket{j}_B \ket{k+n-j}_{E_1'}.
\end{align}
A closed form for $U_{k,n}^{(j)}$ is \cite{Campos1898,Kim2002}:
\begin{align}
U_{k,n}^{(j)} = \sqrt{\frac{k! n!}{j! (k+n-j)!}} \sum_{r=\max(0, j-n)}^{\min(j,k)} (-1)^{n-j+r} \binom{j}{r}\nonumber\\
\times \binom{k+n-j}{k-r} \eta^{\frac{n-j+2r}{2}} (1-\eta)^{\frac{k+j-2r}{2}}
\end{align}
Then the output 3-mode pure state is:
\begin{align}
\ket{\Xi_x}_{BE_1'E_2} = e^{-|\alpha_x|^2/2} &\sum_{k,n=0}^\infty \sum_{j=0}^{k+n} \frac{\alpha_x^k}{\sqrt{k!}} c_n U_{k,n}^{(j)} \nonumber\\
&\times \ket{j}_B \ket{k+n-j}_{E_1'} \ket{n}_{E_2}.
\end{align}
Define the amplitude
\begin{align}
\psi_{j,r,s}^{(x)} = e^{-|\alpha_x|^2/2} \sum_{\substack{k,n \ge 0 \\ k+n = r+s}} &\frac{\alpha_x^k}{\sqrt{k!}} c_n U_{k,n}^{(j)} \nonumber\\
&\times \delta_{r,k+n-j} \delta_{s,n},
\end{align}
which is the coefficient of $\ket{j}_B \ket{r}_{E_1'} \ket{s}_{E_2}$. So the full state is
\begin{align}
\ket{\Xi_x}_{BE_1'E_2} = \sum_{j,r,s} \psi_{j,r,s}^{(x)} \ket{j}_B \ket{r}_{E_1'} \ket{s}_{E_2}.
\end{align}
Tracing out Bob, Eve’s conditional state in the Fock basis is
\begin{align}
\rho_{E_1'E_2|x} = \sum_{r,s,r',s'} \left[ \sum_{j} \psi_{j,r,s}^{(x)} \left( \psi_{j,r',s'}^{(x)} \right)^* \right] \nonumber\\
\times \ket{r,s} \bra{r',s'}_{E_1'E_2}.
\end{align}
So the matrix elements are:
\begin{align}
\left[ \rho_{E_1'E_2|x} \right]_{(r,s),(r',s')} = \sum_{j=0}^{j_\text{max}} \psi_{j,r,s}^{(x)} \left( \psi_{j,r',s'}^{(x)} \right)^*
\end{align}
In the pure-loss case, the BPSK CQ state in reverse reconciliation is
\begin{align}
\rho_{YE} = \frac{1}{2} \sum_{y=0,1} \ket{y}\bra{y}_Y \otimes \rho_{E|y},
\label{CQbpskThermal}
\end{align}
with
\begin{align}
\rho_{E|y} = \sum_{x=0,1} p(x|y)\, \rho_{E|x}.
\end{align}
With excess noise and Eve holding $E_1E_2$, we can now update to:
\begin{align}
\rho_{YE_1E_2} = \frac{1}{2} \sum_{y=0,1} \ket{y}\bra{y}_Y \otimes \rho_{E_1E_2|y}
\label{rhoYEnoiseB}
\end{align}
where
\begin{align}
\rho_{E_1E_2|y} = \sum_{x=0,1} p(x|y)\, \rho_{E_1E_2|x}.
\end{align}
Finally turning the attention to QPSK case, the exact analogue of \eqref{CQbpskThermal} is:
\begin{align}
\rho_{YE_1E_2} = \frac{1}{4} \sum_{y=0}^{3} \ket{y}\bra{y}_Y \otimes \rho_{E_1E_2|y}
\label{rhoYEnoiseQ}
\end{align}
with
\begin{align}
\rho_{E_1E_2|y} = \sum_{k=0}^{3} p(k|y)\, \rho_{E_1E_2|k}.
\end{align}
Thus these CQ states can be used to compute any Petz or sandwiched Rényi conditional entropy $H_a(Y|E_1E_2)$, just as in the pure-loss case, now accounting for the enlarged environment.

\section{Calculation of weight $\omega$ for dimension reduction method}
\label{sec:weight}

In the context of finite-size security for DM CV QKD protocols, the parameter $\omega$ (often referred to as the weight) is essential for the dimension reduction method. Since the Hilbert spaces in continuous-variable systems are infinite-dimensional, the security analysis requires truncating the space to a finite-dimensional subspace $\mathcal{H}^{n_c}$, defined up to a maximum photon number $n_c$ (the cutoff). The parameter $\omega$ represents the probability for the quantum state falling outside this cutoff subspace. Mathematically, if $\Pi^{\perp}$ is the projector onto the orthogonal complement of the cutoff space (the space containing all number states with $n > n_c$), then $\omega$ is defined as the expectation value of this projector on the state $\rho$:
\begin{equation}
    \omega = \mathrm{tr}
     \left[
    \rho \Pi^{\perp}
     \right] = 
    \mathrm{tr} \left[
    \rho \sum_{k=n_c+1}^{\infty} \ket{k}\bra{k}
     \right].
\end{equation}
This value quantifies the probability that the system is found in a state with a photon number strictly greater than $n_c$.

We indicate as $\omega_{exp}$ the value of the expected weight outside the cutoff space, where the state received by Bob is  as a displaced thermal state due to channel noise. We assume Bob's state $\rho_B$ is a thermal state with mean photon number $\bar{n}_B$, expressed in the Fock basis as:
\begin{equation}
    \rho_B = \frac{1}{1+\bar{n}_B} \sum_{n=0}^\infty \left( \frac{\bar{n}_B}{1+\bar{n}_B} \right)^n \ket{n}\bra{n}.
\end{equation}
To calculate $\omega_{exp}$, we compute the trace of this state over the subspace orthogonal to the cutoff. The projector is defined as:
\begin{equation}
    \Pi^{\perp} = \sum_{k=n_c+1}^{\infty} \ket{k}\bra{k}.
\end{equation}
The expected value is given by:
\begin{equation}
    \omega_{exp} = \mathrm{tr}[\rho_B \Pi^{\perp}] = \sum_{k=n_c+1}^{\infty} \bra{k} \rho_B \ket{k}.
\end{equation}
Substituting the expression for $\rho_B$ into the trace yields:
\begin{equation}
    \omega_{exp} = \sum_{k=n_c+1}^{\infty} \bra{k} \left[ \frac{1}{1+\bar{n}_B} \sum_{n=0}^\infty \left( \frac{\bar{n}_B}{1+\bar{n}_B} \right)^n \ket{n}\bra{n} \right] \ket{k}.
\end{equation}
Using the orthonormality of the Fock basis ($\braket{k|n} = \delta_{kn}$), the double summation simplifies to a single sum over the diagonal elements:
\begin{equation}
    \omega_{exp} = \frac{1}{1+\bar{n}_B} \sum_{n=n_c+1}^{\infty} \left( \frac{\bar{n}_B}{1+\bar{n}_B} \right)^n.
\end{equation}
To simplify the notation, let $x = \frac{\bar{n}_B}{1+\bar{n}_B}$. Since $\bar{n}_B > 0$, it follows that $0 < x < 1$. The term outside the sum can be written as $(1-x)$. The expression becomes:
\begin{equation}
    \omega_{exp} = (1-x) \sum_{n=n_c+1}^{\infty} x^n.
\end{equation}
This is a geometric series. The sum of a geometric series starting from an index $N$ is given by $\sum_{n=N}^{\infty} x^n = \frac{x^N}{1-x}$ for $|x| < 1$. In our case, the starting index is $N = n_c + 1$. Substituting this into the formula for $\omega_{exp}$:
\begin{equation}
    \omega_{exp} = (1-x) \cdot \frac{x^{n_c+1}}{1-x} = x^{n_c+1}.
\end{equation}
Finally, substituting back $x = \frac{\bar{n}_B}{1+\bar{n}_B}$ and using the relation $\bar{n}_B = \frac{\eta\xi}{2}$, we obtain the final expression for the expected weight outside the cutoff:
\begin{equation}
    \omega_{exp} = \left( 
    \frac{\bar{n}_B}{1+\bar{n}_B} \right)^{n_c+1} = 
    \left( \frac{\frac{\eta\xi}{2}}{1+\frac{\eta\xi}{2}} 
    \right)^{n_c+1} 
    = 
    \left( \frac{\eta\xi}{2+\eta\xi} 
    \right)^{n_c+1} 
\end{equation}
or alternatively in terms of $\mu$:
\begin{align}
    \omega_{exp} = \left[ 
    \frac{(1-\eta)\mu}{1+(1-\eta)\mu} \right]
    ^{n_c+1}
\end{align}
and recall that $n_c +1 = dim (E)$. The coefficient that enters the security proof is:
\begin{align}
    \Delta (\omega) = 
    \sqrt{\omega}\log_2 (|Y|) 
    +(1+\sqrt{\omega})
    h_2
    \left(
    \frac{\sqrt{\omega}}{1+\sqrt{\omega}}
    \right) 
\end{align}
Thus we have 
\begin{itemize}
    \item BPSK ($|Y| =2$): 
    \begin{align}
         \Delta (\omega) = 
    \sqrt{\omega}
    +(1+\sqrt{\omega})
    h_2
    \left(
    \frac{\sqrt{\omega}}{1+\sqrt{\omega}}
    \right) 
    \end{align}
     \item QPSK ($|Y| =4$): 
    \begin{align}
         \Delta (\omega) = 
    2\sqrt{\omega}
    +(1+\sqrt{\omega})
    h_2
    \left(
    \frac{\sqrt{\omega}}{1+\sqrt{\omega}}
    \right) 
    \end{align}
\end{itemize}

\end{document}